\documentclass[11pt,oneside,a4paper]{article}

\usepackage{amsmath,amssymb,amsthm,amsfonts}
\usepackage[caption=false]{subfig}  % font=normalsize
\usepackage{graphicx}
\usepackage[usenames]{color}

\usepackage{multicol}
\usepackage{syntax}
\usepackage{booktabs}
\usepackage[normalem]{ulem}

\usepackage{latexsym}
\usepackage{stmaryrd}

\usepackage{pstricks}
\usepackage{pst-node}
\usepackage{pst-tree}

\usepackage{paralist}
\usepackage{url}

\usepackage{algorithmic}

\theoremstyle{plain}
\newtheorem{theorem}{Theorem}
\newtheorem{lemma}{Lemma}
\newtheorem{proposition}{Proposition}
\newtheorem{corollary}{Corollary}
\newtheorem{remark}{Remark}

\theoremstyle{definition}
\newtheorem{definition}{Definition} 
\newtheorem{example}{Example}

\newcommand{\T}{\mathcal{T}}
\newcommand{\C}{\mathcal{C}}
\newcommand{\D}{\mathbf{D}}
\renewcommand{\phi}{\varphi}
\newcommand{\tuple}[1]{{\langle#1\rangle}}

\def\punto{$\hspace*{\fill}\Box$}

\newcommand{\cust}{\mbox{Cust}}
\newcommand{\ord}{\mbox{Ord}}

\newcommand{\lin}{\mbox{Item}}

\newcommand{\cname}{name}

\newcommand{\okey}{okey}
\newcommand{\ckey}{ckey}

\newcommand{\nop}[1]{}

\title{Factorised Representations of Query Results\footnote{A preliminary version has been submitted for publication on March 1, 2011.}}
\author{Dan Olteanu and Jakub Z\'{a}vodn\'{y}\\
Computing Laboratory, University of Oxford\\
Wolfson Building, Parks Road, OX1 3QD, Oxford, UK}

\date{}
\begin{document}

\maketitle

\begin{abstract}
Query tractability has been traditionally defined as a function of
input database and query sizes, or of both input and output sizes,
where the query result is represented as a bag of tuples. In this
report, we introduce a framework that allows to investigate
tractability beyond this setting. The key insight is that, although
the cardinality of a query result can be exponential, its structure
can be very regular and thus factorisable into a nested representation
whose size is only polynomial in the size of both the input database
and query.

For a given query result, there may be several equivalent
representations, and we quantify the regularity of the result by its
{\em readability}, which is the minimum over all its representations
of the maximum number of occurrences of any tuple in that
representation. We give a characterisation of select-project-join
queries based on the bounds on readability of their results for any
input database. We complement it with an algorithm that can find
asymptotically optimal upper bounds and corresponding factorised
representations.
\end{abstract}

\section{Introduction}
\label{sec:introduction}

This paper studies properties related to the representation of results of select-project-join queries under bag semantics. In approaching this challenge, we depart from the standard flat representation of query results as bags of tuples and consider nested representations of query results that can be exponentially more succinct than a mere enumeration of the result tuples. The relationship between a flat representation and a nested, or {\em factorised}, representation is on a par with the relationship between logic functions in disjunctive normal form and their equivalent nested forms obtained by algebraic factorisation. When compared to flat representations of query results, factorised representations are both succinct and informative.

\begin{figure}[t]
\begin{center}
\begin{small}
\begin{tabular}{@{~}c@{\quad}c@{\quad}c@{\quad}c@{~}}
\begin{tabular}{@{}c@{~}||@{~}c@{~}@{~}c@{~}@{~}}
Cust & ckey & name  \\
\hline
$c_1$ & 1     & Joe    \\
$c_2$ & 2     & Dan    \\
$c_3$ & 3     & Li     \\
$c_4$ & 4     & Mo     
\end{tabular}
&
\begin{tabular}{@{~}c@{~}||@{~}c@{~}c@{~}c@{~}}
Ord & ckey & okey & date\\
\hline
$o_1$ & 1     & 1     & 1995\\
$o_2$ & 1     & 2     & 1996\\
$o_3$ & 2     & 3     & 1994\\
$o_4$ & 2     & 4     & 1993\\
$o_5$ & 3     & 5     & 1995\\
$o_6$ & 3     & 6     & 1996
\end{tabular}
&
\begin{tabular}{@{~}c@{~}||@{~}c@{~}c@{~}}
Item & okey & disc \\
\hline
$i_1$ & 1     & 0.1\\
$i_2$ & 1     & 0.2\\
$i_3$ & 3     & 0.4\\
$i_4$ & 3     & 0.1\\
$i_5$ & 4     & 0.4\\
$i_6$ & 5     & 0.1
\end{tabular}
\end{tabular}
\end{small}
\end{center}\vspace*{-1em}

\caption{A TPC-H-like database.}
\label{fig:db}
\vspace{-1em}
\end{figure}

\begin{example}\label{ex:f-representation}
Consider a simplified TPC-H scenario with customers, orders, and discounted line items, as depicted in Figure~\ref{fig:db}. Each tuple is annotated with an identifier. The query $\cust\Join_{\ckey}\ord\Join_{\okey}\lin$ reports all customers together with their orders and line items per order. A flat representation of the result is presented below:
\begin{center}
\begin{small}
\begin{tabular}{@{}c@{~}||@{~}c@{~}c@{~}c@{~}c@{~}c@{~}c@{~}}
 $Q$ & ckey & name & okey & date &  disc \\\hline
 $c_1o_1i_1$ & 1     & Joe &  1  & 1995 & 0.1  \\
 $c_1o_1i_2$ & 1     & Joe &  1  & 1995 & 0.2  \\
 $c_2o_3i_3$ & 2     & Dan &  3  & 1994 & 0.4  \\
 $c_2o_3i_4$ & 2     & Dan &  3  & 1994 & 0.1  \\
 $c_2o_4i_5$ & 2     & Dan &  4  & 1993 & 0.4 \\
 $c_3o_5i_6$ & 3     & Li  &  5  & 1995 & 0.1
\end{tabular}
\end{small}
\end{center}
For each result tuple, the identifiers of tuples that contributed to it are shown. For instance, the input tuples with identifiers $c_1$, $o_1$, and $i_1$ contribute to the first result tuple. Our factorised representation is based on an algebraic factorisation of a polynomial that encodes the result. This encoding is constructed as follows. Each result tuple is annotated with a product of identifiers of tuples contributing to it. The whole result is then a sum of such products. For this example, the sum of products of identifiers is:
\begin{small}$$\psi_1 = c_1o_1i_1+c_1o_1i_2+ c_2o_3i_3+c_2o_3i_4+c_2o_4i_5+
c_3o_5i_6.$$\end{small}
An equivalent nested expression would be:
\begin{small}$$\psi_2 = c_1o_1(i_1+i_2)+ c_2(o_3(i_3+i_4)+o_4i_5)+ c_3o_5i_6.$$\end{small}
A factorised representation of the result is an extension
of this nested expression with values from the result tuples:
\begin{small}
\begin{align*}
& {\color{blue}c_1}\tuple{1,Joe}{\color{blue}o_1}\tuple{1,\mbox{1995}}{\color{blue}(i_1}\tuple{0.1}+{\color{blue}i_2}\tuple{0.2}{\color{blue})}+\\
& {\color{blue}c_2}\tuple{2,Dan}{\color{blue}(o_3}\tuple{3,\mbox{1994}}{\color{blue}(i_3}\tuple{0.4}+{\color{blue}i_4}\tuple{0.1}{\color{blue})}+
  {\color{blue}o_4}\tuple{4,\mbox{1993}}{\color{blue}i_5}\tuple{0.4}{\color{blue})}+\\
& {\color{blue}c_3}\tuple{3,Li}{\color{blue}o_5}\tuple{5,\mbox{1995}}{\color{blue}i_6}\tuple{0.1}.
\end{align*}
\end{small}
To correctly interpret this representation as a relation, we also need
a mapping of identifiers to schemas. For instance, the identifiers
$c_1$ to $c_3$ are mapped to $(\ckey,\cname)$, which serves as schema
for tuples $\tuple{1,Joe}$, $\tuple{2,Dan}$, and $\tuple{3,Li}$.\punto
\end{example}

We can easily recover the result tuples from the factorised
representation with {\em polynomial delay}, i.e., the delay between
two successive tuples is polynomial in the size of the representation.
For this, consider the parse tree of the representation. The inner
nodes stand for product or sum, and the leaves for identifiers with
tuples. A result tuple is a concatenation of the tuples at the leaves
after choosing one child for each sum and all children for each
product. We assume here that from a user perspective, iterating over
the result with small delay is more important than presenting the
whole result at once.

Factorised representations can be more informative than flat
representations in that they better explain the result and spell out
the extent to which certain input fields contribute to result tuples
either individually or in groups with other fields. This enables a
shift in the presentation of the result from a tuple-by-tuple view to
a {\em kernel} view, in which commonalities across result tuples are
made explicit by exploiting the factorised representation. We can
depict it graphically as its parse tree or textually as a
serialisation of this tree in tabular form. 

\begin{example}
The textual presentation of our factorised representation in
Example~\ref{ex:f-representation} could be the left one below:
\begin{center}
\begin{small}
\begin{tabular}{@{~}c@{~}c@{~}c@{~}c@{~}c@{~}c@{~}}
ckey & name & okey & date & disc \\\hline
1    & Joe &  1  & 1995 & 0.1  \\
     &     &     &      & 0.2\\
2    & Dan & 3   & 1994 & 0.4\\
     &     &     &      & 0.1\\
     &     & 4   & 1993 & 0.4\\
3    & Li  & 5   & 1995 & 0.1
\end{tabular}%
\hspace*{2em}
\begin{tabular}{@{~}ccc@{~}}
name &    & items \\\hline
     &  & \\
\cline{1-1} \cline{3-3}  Joe  &    & LCD  \\
Dan  & \texttt{x}   & LED \\
\cline{3-3} Li   &    & \\
\cline{1-1}      &    & \\
Mo   &    & BW
\end{tabular}
\end{small}
\end{center}
It is easy to see that two discounted line items (with discount 0.1
and 0.2) are for the same order 1 of customer Joe.

Consider now the following factorised representation
\begin{small}
\begin{align*}
& {\color{blue}(s_1}\tuple{Joe}+{\color{blue}s_2}\tuple{Dan}+
{\color{blue}s_3}\tuple{Li}
{\color{blue})(p_1}\tuple{LCD}+{\color{blue}p_2}\tuple{LED}{\color{blue})}+\\
& {\color{blue}s_4}\tuple{Mo}{\color{blue}p_3}\tuple{BW}
\end{align*}
\end{small}
where $s_1$ to $s_4$ identify suppliers, and $p_1$ to $p_3$ identify
items. This representation encodes that Joe, Dan, and Li supply both
LCD and LED TV sets, and Mo supplies BW TV sets. A textual
presentation of this result could be the right one above. The blocks
between the horizontal lines encode tuples obtained by combining any
of the names with any of the items.  This relational product is
suggested by the \texttt{x} symbol between the blocks. (We skip the
details on the mapping between the parse trees of factorised
expressions and their tabular presentations.)\punto
\end{example}

In the factorised representation $\psi_2$ and in contrast to its
equivalent flat representation $\psi_1$, each identifier only occurs
once.  We seek good factorised representations of a query result in
which each identifier occurs a small number of times. The maximum
number of occurrences of any identifier in a representation, or in any
of its equivalent representations, defines the {\em readability} of
that representation. Readability implies bounds on the representation
size. In our example, the size of the factorised representation is at
most linear in the size of the input database, since its readability
is one.

Our study of readability is with respect to tuple identifiers and
aligns well with query evaluation under bag semantics. This is
different from readability with respect to values. For instance,
$\psi_2$ has readability one, yet a value may occur several times in
the tuples of $\psi_2$, e.g., the discount value of 0.1. Studying
readability with respect to values is especially relevant to query
evaluation under set semantics.

\section{Contributions}

The main contributions of this paper are as follows.
%\begin{compactitem}
\begin{itemize}

\item We introduce factorised representations, a succinct and complete 
representation system for (results of queries in) relational
databases. In contrast to the standard tabular representation of a bag
of tuples, factorised representations can be exponentially more
succinct by factoring out commonalities across tuples. They also allow
for an intuitive presentation, whereby commonalities across tuples are
made explicit.

\item We give lower and upper bounds on the readability of basic queries 
with equality or inequality joins.

\end{itemize}
%\end{compactitem}

\noindent The following holds for select-project-join queries with 
equality joins.

%\begin{compactitem}
\begin{itemize}
\item We introduce factorisation trees that define generic classes of 
factorised representations for query results. Such trees are
statically inferred from the query and are independent of the database
instance. A factorised representation $\Phi(\T)$ modelled on $\T$ has
the nesting structure of $\T$ for any input database.

\item We give a tight characterisation of queries based on their 
readability with respect to factorisation trees. For any query $Q$, we
can find a rational number $f(Q)$ such that the readability of $Q(\D)$
is at most $|Q|\cdot|\D|^{f(Q)}$ for any database $\D$, while for any
factorisation tree $\T$ there exist databases for which the
factorisation of $Q(\D)$ modelled on $\T$ has at least
$(|\D|/|Q|)^{f(Q)}$ occurrences of some identifier.

\item For any query $Q$, we present an algorithm that iterates over the 
factorisation trees of $Q$ and finds an optimal one $\T$. Given $\T$,
we present a second algorithm that computes in time $O(|Q|\cdot
|\D|^{f(Q)+1})$ for any database $\D$ a factorised representation
$\Phi(\T)$ of $Q(\D)$ with readability at most $|Q|\cdot |\D|^{f(Q)}$
and at most $|\D|^{f(Q)+1}$ occurrences of identifiers.

\item Our characterisation captures as a special case the known class of 
hierarchical non-repeating queries~\cite{dalvi07efficient} that have
readability one~\cite{OH2008}. We also show that non-hierarchical
non-repeating queries have readability $\Omega(\sqrt{|\mathbf{D}|})$
for arbitrarily large databases $\D$.
\end{itemize}
%\end{compactitem}

Section~\ref{sec:constants} shows how to extend the above results to
selections that contain equalities with constants. Proofs are deferred
to the appendix.

\section{Related Work}
\label{sec:relatedwork}

Our study has strong connections to work on readability of Boolean
functions, provenance and probabilistic databases, streamed query
evaluation, syntactic characterisations of queries with polynomial
time combined complexity or polynomial output size, and selectivity
estimation in relational engines. The present work is nevertheless
unique in its use of succinct nested representations of query results.

The notion of readability is borrowed from earlier work on Boolean
functions, e.g.,~\cite{Golumbic06a,Golumbic08,Elbassioni09}. Like in
our case, a formula $\Phi$ is \emph{read-m} if each variable appears
at most $m$ times in $\Phi$, and the readability of a formula or a
function $\Phi$ is the smallest number $m$ such that there is a
read-$m$ formula equivalent to $\Phi$. Checking whether a monotone
function in disjunctive normal form has readability $m=1$ can be done
in time linear in both the number of terms and number of
variables~\cite{Golumbic08}. This problem is open for $m=2$, and
already hard for $m>2$ or for $m=2$ and monotone nested
functions~\cite{Elbassioni09}. This strand of work differs from ours
in two key points. Firstly, we only consider \emph{algebraic}, and not
Boolean, equivalence; in particular, idempotence ($x\cdot x = x$) is
not considered since a reduction in the arity of any product in the
representation would violate the mapping between tuple fields and
schemas. Secondly, we only consider functions/formulas arising as
results of queries, and classify queries based on worst-case analysis
of the readability of their results.

The hierarchical property~\cite{dalvi07efficient} of queries plays a
central role in studies with seemingly disparate focus, including the
present one, probabilistic databases, and streamed query
evaluation. Our characterisation of query readability essentially
revolves around how far the query is from its hierarchical
subqueries. We show that, within the class of queries without
repeating relation symbols, the readability of any non-hierarchical
query is dependent on the size of the input database, while for any
hierarchical query, the readability is always one. This latter result
draws on earlier work in the context of probabilistic
databases~\cite{OH2008,OHK2009,FinkOlteanu:ICDT:2011}, where read-once
polynomials over random variables are useful since their exact
probability can be computed in polynomial time. Read-$m$ functions for
$m > 2$ are of no use in probabilistic databases, since probability
computation for such functions over random variables is
\#P-hard~\cite{Vadhan2001}. In our case, however, readability
polynomial in the sizes of the input database and query is acceptable,
since it means that the size of the result representation is
polynomial, too.

Mirroring the dichotomies in the probabilistic and query readability
contexts, it has been recently shown that the hierarchical property
divides queries that can be evaluated in one pass from those that
cannot in the finite cursor machine model of
computation~\cite{Grohe:TCS:2009}. In this model, queries are
evaluated by first sorting each relation, followed by one pass over
each relation. It would be interesting to investigate the relationship
between the readability of a query $Q$ and the number of passes
necessary in this model to evaluate $Q$.

Our study fits naturally in the context of provenance
management~\cite{Green:PODS:2007}. Indeed, the polynomials over tuple
identifiers discussed in Example~\ref{ex:f-representation} are
provenance polynomials and nested representations are algebraic
factorisations of such polynomials. In this sense, our work
contributes a characterisation of queries by readability and size of
their provenance polynomials.

Earlier work in incomplete databases has introduced a representation
system called world-set decompositions~\cite{OKA08gWSD} to represent
succinctly sets of possible worlds. Such decompositions can be seen as
factorised representations whose structure is a product of sums of
products.

There exist characterisations of conjunctive queries with polynomial
{\em time} combined complexity~\cite{AHV95}. The bulk of such
characterisations is for various classes of Boolean queries under set
semantics. In this context, even simple non-Boolean conjunctive
queries such as a product of $n$ relations would require evaluation
time exponential in $n$. Our approach exposes the simplicity of this
query, since its readability is one and the smallest factorised
representation of its result has linear size only and can be computed
in linear time. Factorised representations could thus lead to larger
classes of tractable queries.

Finally, there has been work on deriving bounds on the cardinality of
query results in terms of structural properties of
queries~\cite{Gottlob99,AGM08,Gottlob09a}. Our work uses the results
in \cite{AGM08} and quantifies how much they can be improved due to
factorised representations.

\section{Preliminaries}

{\noindent\bf Databases.} We consider relational databases as
collections of annotated
\emph{relation instances}, as in
Example~\ref{ex:f-representation}. Each relation instance $\mathbf{R}$
is a bag of tuples in which each tuple is annotated by an
identifier. We denote by $\mathcal{I}(\mathbf{R})$ the set of
identifiers in $\mathbf{R}$, by $\mathcal{S}(\mathbf{R})$ the schema
of $\mathbf{R}$, and call the pair
$(\mathcal{I}(\mathbf{R}),\mathcal{S}(\mathbf{R}))$ its
\emph{signature}.

The \emph{size} of a relation instance $\mathbf{R}$ is the number of
tuples in $\mathbf{R}$, denoted by $|\mathbf{R}|$. The number of
distinct tuples in $\mathbf{R}$ is denoted by $||\mathbf{R}||$. The
size $|\D|$ of a database $\D$ is the total number of tuples in all
relations of $\D$.

\begin{remark}\em
For the purpose of analysing the complexity of our algorithms, we
assume that the tuples in the input database are of constant size. In
many scenarios, this is however not realistic since even the encodings
of the tuple identifiers must have size at least logarithmic in
$\D$. If the maximal size of a tuple in $\D$ is $C(\D)$, the time
complexity increases by an additional factor $C(\D)$ or similar,
depending on the exact computation model used.\punto
\end{remark}

{\noindent\bf Queries.} We consider conjunctive or select-project-join
queries written in relational algebra but with evaluation under bag
semantics. Such queries have the form
$\pi_{\bar{A}}(\sigma_{\phi}(R_1\times\ldots\times R_n))$, where
$R_1,\ldots,R_n$ are relations, $\phi$ is a conjunction of equalities
of the form $A_1=A_2$ with attributes $A_1$ and $A_2$, and $\bar{A}$
is a list of attributes of relations $R_1$ to $R_n$. The size $|Q|$ of
the query $Q$ is the total number of relations and attributes in $Q$.

Let $Q = \pi_{\bar{A}} ( \sigma_\phi (R_1 \times \dots \times R_n))$
be a query and $\D$ be a database containing a relation instance
$\mathbf{R}_i$ of the correct schema for each relation $R_i$ in
$Q$. The result $Q(\D)$ of the query $Q$ on the database $\D$ is a
relation instance whose tuples are exactly those
$\pi_{\bar{A}}(t_1\times\dots\times t_n)$ for which $t_i \in
\mathbf{R}_i$ and $t_1\times\dots\times t_n \models \phi$. The tuple
$\pi_{\bar{A}}(t_1\times\dots\times t_n)$ is annotated by
$id_1id_2\dots id_n$, where $id_i$ is the identifier of $t_i$ in
$\mathbf{R}_i$.

Every query can be brought into an equivalent form where all relations
as well as all their attributes are distinct. To recover the original
query $Q_0$ from the rewritten one $Q$, we keep a function $\mu$ that
maps the relations in $Q$ to relations in $Q_0$, and the attributes of
$R$ in $Q$ to those of $\mu(R)$ in $Q_0$. For technical reasons, we
will only consider the rewritten queries in further text, the mapping
$\mu$ will carry the information about different relation symbols
representing the same relation. If a query $Q$ has two relations with
the same mapping $\mu(R)$, then $Q$ is \emph{repeating}; otherwise,
$Q$ is \emph{non-repeating}.

For any attribute $A$, let $A^*$ be its equivalence class, that is,
the set of all attributes that are transitively equal to $A$ in
$\phi$, and let $r(A)$ be the set of relations that have attributes in
$A^*$.

A query is {\em hierarchical}\footnote{The original
definition~\cite{dalvi07efficient} does not consider the output
attributes $\bar{A}$ when checking the hierarchical property.}, if for
any two attributes $A$ and $B$, either $r(A)\subseteq r(B)$, or
$r(A)\supset r(B)$, or $r(A)\cap r(B)=\emptyset$.

\begin{example}\label{ex:hierarchical}
The query from Example~\ref{ex:f-representation} in the introduction
is non-repeating and not hierarchical.

Consider the relations $R$, $S$, and $T$ over schemas $\{A_R\}$,
$\{A_S,B_S\}$, and $\{B_T,U\}$ respectively. The query
$\pi_{\bar{A}}[\sigma_{A_R=A_S,B_S=B_T} (R\times S\times T)]$ is not
hierarchical (independently of the set $\bar{A}$), since
$r(A_S)\not\subseteq r(B_S)$, $r(A_S)\not\supset r(B_S)$, but
$r(A_S)\cap r(B_S)=\{S\}$. The query
$\pi_{\bar{A}}[\sigma_{A_R=A_S,B_S=B_T,A_R=U} (R\times S\times T)]$,
equivalent to $R(A),S(A,B),T(B,A)$, is hierarchical, since
$r(A_R)=r(A_S)=r(U)=\{R,S,T\}$ $\supset r(B_S)=r(B_T)=\{S,T\}$.\punto
\end{example}

\section{Factorised Representations}
\label{sec:factorised-representations}

In this section we formalise the notion of factorised representations,
their algebraic equivalence, and readability. We also give tight
bounds on the readability of certain factorised representations that
are used in the next sections to derive bounds on the readability of
query results.

\begin{definition}
A {\em factorised representation}, or f-representa\-ti\-on for short, $\Phi$
over a set of signatures $\mathrm{Sign}$ is
\begin{compactitem}
\item $\Phi_1+\cdots+\Phi_n$, where $\Phi_1$ to $\Phi_n$ are f-representations 
over Sign, or

\item $\Phi_1\cdots\Phi_n$, where $\Phi_1$ to $\Phi_n$ are f-representations over $\mathrm{Sign}_1$ 
to $\mathrm{Sign}_n$, respectively, and these signatures form a
disjoint cover of $\mathrm{Sign}$, or

\item $id\langle t \rangle$, where $id\in{\cal R}_i$ and $t$ is a tuple over schema 
$\mathcal{S}_i$, and
$\mathrm{Sign}=\{(\mathcal{R}_i,\mathcal{S}_i)\}$.
\end{compactitem}
The {\em polynomial} of $\Phi$ is $\Phi$ without tuples on
identifiers. The \emph{size} of (the polynomial of) $\Phi$ is the
total number of occurrences of identifiers in $\Phi$.\punto
\end{definition}

Two examples of f-representations are given in
Section~\ref{sec:introduction}. A relational database can have several
algebraically equivalent f-representations, in the sense that these
f-representations represent the same tuples and
polynomials. Syntactically, we define equivalence of f-representations
as follows.

\begin{definition}
  Two f-representations are {\em equivalent} if one can be obtained
  from the other using distributivity of product over sum and
  commutativity of product and sum.\punto
\end{definition}

Each f-representation has an equivalent {\em flat} f-representation,
which is a sum of products. A product $i_1\tuple{t_1}\cdots
i_n\tuple{t_n}$ defines the tuple $\tuple{t_1\circ\cdots\circ t_n}$
over schema $\bigcup_i \mathcal{S}_i$, which is a concatenation of
tuples $\tuple{t_1}$ to $\tuple{t_n}$, and is annotated by the product
$i_1\dots i_n$.

\begin{definition}
\label{def:facrep1}
The relation encoded by an f-representation $\Phi$ consists of all
tuples defined by the products in the flat f-representation equivalent
to $\Phi$.\punto
\end{definition}

Since flat f-representations are standard relational databases annotated
with identifiers, it means that any relational database can be encoded
as an f-representation. This property is called completeness.
\begin{proposition}
  Factorised representations form a complete representation system for
  relational data.
\end{proposition}
In particular, this means that there are f-representations of the
result of any query in a relational database.

\begin{definition}
\label{def:facrep}
Let $Q = \pi_{\bar{A}} ( \sigma_\phi (R_1 \times \dots \times R_n))$
be a query, and $\mathbf{D}$ be a database.
An f-representation $\Phi$ encodes the result $Q(\mathbf{D})$ if its
equivalent flat f-representation contains exactly those products
$id_1\tuple{\pi_{\bar{A}}(t_1)} \cdot \ldots \cdot
id_n\tuple{\pi_{\bar{A}}(t_n)}$ for which
$\pi_{\bar{A}}(t_1\times\dots\times t_n)\in Q(\mathbf{D})$, and $id_i$
is the identifier of $t_i$ for all $i$.

The signature set of $\Phi$ consists of the signatures $({\cal
I}_i,{\cal S}_i)$ for each query relation $R_i$, such that ${\cal
I}_i$ is the set of identifiers of the relation instance in $\D$
corresponding to $R_i$, and ${\cal S}_i$ is the schema of $R_i$ in $Q$
restricted to the attributes in $\bar{A}$.\punto
\end{definition}

\nop{Example~\ref{ex:f-representation} shows an f-representation of a query
result.}

Flat f-representa\-ti\-ons can be exponentially less succinct
than equivalent nested f-representa\-ti\-ons, where the exponent is the
size of the schema.

\begin{proposition}
Any flat representation equivalent to the f-representation
$(x_1\alpha+y_1\beta)\cdot \ldots \cdot(x_n\alpha+y_n\beta)$ over the
signatures $(\{x_1,\ldots,x_n\},{\cal A})$ and
$(\{y_1,\ldots,y_n\},{\cal B})$ has size $2^n$.
\end{proposition}

In addition to completeness and succinctness, f-representa\-tions allow for efficient enumeration of their tuples.
\begin{proposition}\label{prop:tuple-enumeration}
The tuples of an f-representation $\Phi$ can be enumerated with $O(|\Phi|\log |\Phi|)$ delay and space.
\end{proposition}
 
Besides the size, a key measure of succinctness of f-representations is their
readability. We extend this notion to query results for any input
database in Section~\ref{sec:query-readability}.

\begin{definition}
An f-representation $\Phi$ is \emph{read-$k$} if the maximum number of
occurrences of any identifier in $\Phi$ is $k$. The \emph{readability}
of $\Phi$ is the smallest number $k$ such that there is a read-$k$
f-representation equivalent to $\Phi$.\punto
\end{definition}
Since the readability of $\Phi$ is the same as of its polynomial, we
will use polynomials of f-representations when reasoning about their
readability.

\begin{example}
In Example~\ref{ex:f-representation}, the polynomial $\psi_1$ is
read-3 and the polynomial $\psi_2$ is read-1. They are equivalent and
hence both have readability one.\punto
\end{example}

Given the readability $\rho$ and the number $n$ of distinct
identifiers of a polynomial, we can immediately derive an upper bound
$n\rho$ on its size. A better upper bound can be obtained by taking
into account the (possibly different) number of occurrences of each
identifier. However, for polynomials of query results, the bound
$n\rho$ is often dominated by the readability $\rho$.

In Section~\ref{sec:query-readability}, we define classes of queries
that admit polynomials of low readability, such as constant
readability. We next give examples of polynomials with readability
depending polynomially on the number of identifiers.

\begin{lemma}\label{lemma:readability-rst}
The polynomial $p_N = \sum_{i,j = 1}^N r_i s_{ij} t_j$ has readability $\frac{N}{2} + O(1)$.
\end{lemma}

Lemma~\ref{lemma:readability-rst} can be generalised as follows.

\begin{theorem}
The readability of the polynomial $p_{N,M} = \sum_{i=1}^N\sum_{j=1}^M r_i s_{ij} t_j$ is $\frac{NM}{N+M} + O(1)$.
\end{theorem}

If we drop the set of identifiers $s_{ij}$, the readability becomes
one. However, if we restrict the relationship between the remaining
identifiers, the readability increases again.

\begin{theorem}\label{th:crown}
The readability of the polynomial $q_N = \sum_{i,j=1; i\neq j}^ N r_it_j$ is \\ $\Omega(\frac{\log{N}}{\log\log{N}})$ and $O(\log{N})$.
\end{theorem}

The polynomials $p_{N,M}$ and $q_N$ are relevant here due to their
connection to queries: $p_{N,M}$ is the polynomial of the query
$\sigma_{\phi}(R\times S\times T)$, where $\phi := (A_R=A_S\wedge
B_S=B_T)$ and the schemas of $R$, $S$, and $T$ are $\{A_R\}$,
$\{A_S,B_S\}$, and $\{B_T\}$ respectively, on the database where
$\mathbf{R}$, $\mathbf{S}$ and $\mathbf{T}$ are full relations with
$|\mathbf{R}|=n$ and $|\mathbf{T}|=m$. Also, $q_N$ is the polynomial
of the disequality query $\sigma_{A_R\neq B_T}(R\times T)$. If $i\neq
j$ is replaced by $i\leq j$ in $q_N$, the lower and upper bounds on
readability on this new polynomial $q'_N$ still hold, and we obtain
the result of an inequality query.

A lower bound of $\sqrt{\frac{\log{N}}{\log\log{N}}}$ on the
readability of $q'_N$ is already known even in the case when Boolean
factorisation is allowed~\cite{Golumbic06a}.

\section{Factorisation Trees}
\label{sec:ftrees}

\begin{figure}[t]
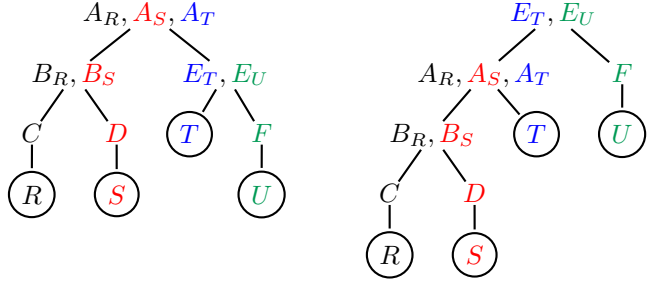

\begin{small}
\[
\psset{levelsep=8mm, nodesep=1pt, treesep=5mm}
\pstree{\TR{A_R,{\color{red}A_S},{\color{blue}A_T}}}
{
  \pstree{\TR{B_R,{\color{red}B_S}}}
  {
   \pstree{\TR{C}}
    {
      \Tcircle{R}
    }
   \pstree{{\color{red}\TR{D}}}
    {
      \Tcircle{{\color{red}S}}
    }
  }
  \pstree{\TR{{\color{blue}E_T},{\color{ForestGreen}E_U}}}
  {
    \Tcircle{{\color{blue}T}}
    \pstree{{\color{ForestGreen}\TR{F}}}
     {
       \Tcircle{{\color{ForestGreen}U}}
     }
  }
}%
\hspace*{3em}%
\psset{levelsep=8mm, nodesep=1pt, treesep=5mm}
\pstree{\TR{{\color{blue}E_T},{\color{ForestGreen}E_U}}}
{
  \pstree{\TR{A_R,{\color{red}A_S},{\color{blue}A_T}}}
  {
   \pstree{\TR{B_R,{\color{red}B_S}}}
    {
     \pstree{\TR{C}}
      {
        \Tcircle{R}
      }
     \pstree{{\color{red}\TR{D}}}
      {
        \Tcircle{{\color{red}S}}
      }
    }
    \Tcircle{{\color{blue}T}}
  }
  \pstree{{\color{ForestGreen}\TR{F}}}
   {
     \Tcircle{{\color{ForestGreen}U}}
   }
}
\]
\end{small}
\caption{F-trees for the query in Example~\ref{ex:f-trees}.}
\label{fig:f-trees}
\end{figure}

We next introduce a generic class of factorised representations for query results, constructed using so-called factorisation trees, whose nesting structure and readability properties can be described statically from the query only. We present an algorithm that, given a factorisation tree $\T$ of a query $Q$, and an input database $\D$, computes a factorised representation of $Q(\D)$, whose nesting structure is that defined by ${\cal T}$. Factorisation trees are used in Section~\ref{sec:query-readability} to obtain bounds on the readability of queries.

\begin{definition}
A \emph{factorisation tree (f-tree)} for a query $Q$ is a rooted unordered forest $\T$, where
\begin{compactitem}
\item there is a one-to-one mapping between inner nodes in $\T$ and equivalence classes of attributes of $Q$,
\item there is a one-to-one mapping between leaf nodes in $\T$ and relations in $Q$, and
\item the attributes of each relation only appear in the ancestors of 
its leaf.\punto
\end{compactitem}
\end{definition}

\begin{example}\label{ex:f-trees}
Consider the relations $R$, $S$, $T$, and $U$ over schemas
$\{A_R,B_R,C\}$, $\{A_S,B_S,D\}$, $\{A_T,E_T\}$, and $\{E_U,F\}$
respectively, and the query $Q = \sigma_{\phi}(R\times S\times T\times
U)$ with $\phi =
(A_R=A_S,A_R=A_T,B_R=B_S,E_T=E_U)$. Figure~\ref{fig:f-trees} depicts
two f-trees for $Q$.

Consider now the query $Q' = \sigma_{\phi}(R\times S\times T)$ with
$\phi = (A_R=A_S,A_R=A_T,B_R=B_S)$. Figure~\ref{fig:r-f-trees} on page~\pageref{fig:r-f-trees} shows
two f-trees for $Q'$ as well as a partial tree that cannot be extended
to an f-tree since the attributes $A_S$ and $D$ of $S$ lie in
different branches.\punto
\end{example}

Each f-tree for $Q$ is a recipe for producing an f-representation of
the result $Q(\D)$ for any database $\D$. For a given query $Q$ and
database $\D$, this f-representation is called the $\T$-factorisation
of $Q(\D)$ and is denoted by $\Phi(\T)$. Figure~\ref{fig:algorithm1}
gives a recursive function $\llbracket\cdot\rrbracket$ that computes
the $\T$-factorisation of $Q(\D)$. A more detailed implementation of
this function, including an analysis of its time and space complexity,
is given in Section~\ref{sec:falgorithm}.

\begin{figure}[t!]
\framebox[\columnwidth]{
\parbox{8.7cm}{
\begin{align*}
\llbracket \mbox{siblings } \{{\cal T}_1,\ldots,{\cal T}_n\} \rrbracket(\gamma) &= 
\llbracket {\cal T}_1\rrbracket(\gamma)\cdots\llbracket {\cal T}_n\rrbracket(\gamma)\\
\llbracket \mbox{inner node } A^*(\mathcal{U}) \rrbracket(\gamma) &= 
\sum_{a \in \mbox{Dom}_{A^*}} (\llbracket \mathcal{U}\rrbracket(\gamma,A^* = a))\\
\llbracket \mbox{leaf } R \rrbracket(\gamma) &= 
\sum_{t_j \in \sigma_\gamma(\mathbf{R})} id_j\tuple{\pi_{\mbox{head}(Q)}(t_j)}
\end{align*}
}}
\caption{The $\T$-factorisation of a query result $Q(\mathbf{D})$ is computed as $\Phi(\T)=\llbracket\T\rrbracket(\top)$, where $\top$ is the constant true (an empty conjunction). For a relation $R$ in $Q$, $\mathbf{R}$ is the corresponding relation instance in the input database $\D$.}
\label{fig:algorithm1}
\end{figure}

The function $\llbracket\cdot\rrbracket$ recurses on the structure of
$\T$. The parameter $\gamma$ is a conjunction of equality conditions
that are collected while traversing the f-tree top-down. Initially,
$\gamma$ is an empty conjunction $\top$. In case $\T$ is a forest
 $\{\T_1,\ldots,\T_n\}$, we return the f-representation defined by the
product of f-representations of each tree in $\T$. If $\T$ is single tree $A^*(\mathcal{U})$ with root $A^*$ and children $\mathcal{U}$, we return the f-representation of a sum over all possible domain values $a$ of the attributes in $A^*$ of the f-representations of the children $\mathcal{U}$. To compute these, for each possible value $a$ we simply recurse on $\mathcal{U}$, appending to $\gamma$ the equality condition $A^* = a$. Finally, in case $\T$ is a leaf $R$, we return a sum of f-representations for result tuples in $R$, that is,
only those tuples that satisfy $\gamma$. (When evaluating the
selection with $\gamma$ on $R$, we only consider the equalities on
attributes of $R$.) In the f-representation we only include attributes
from $Q$'s projection list, along with the tuple identifier.

The symbolic products and sums in Figure~\ref{fig:algorithm1} are of course
expanded out to produce a valid f-representation. However, we will
often keep the sums symbolic, abbreviate $\sum_{a\in\mbox{Dom}_{A^*}}$
to $\sum_{A^*}$ and write $R$ instead of $\sum_{t_j \in
\sigma_\gamma(\mathbf{R})} id_j\tuple{\pi_{\mbox{head}(Q)}(t_j)}$ for
the expression generated by the leaves. The condition $\gamma$ can be
inferred from the position in the expression, so we can still recover
the original representation and write out the sums explicitly. Such an
abbreviated form is independent of the database $\D$ and conveniently reveals the structure of any $\T$-factorisation.

\begin{figure}[t]
\begin{center}
\begin{small}
\begin{tabular}{c@{\quad}c@{\quad}c@{\quad}c}
\begin{tabular}{@{}c@{~}||@{~}c@{~}@{~}c@{~}@{~}c@{~}@{~}}
 $R$ & $A_R$ & $B_R$ & $C$ \\
\hline
 $r_{111}$ & 1 & 1 & 1 \\
 $r_{122}$ & 1 & 2 & 2 \\
 $r_{212}$ & 2 & 1 & 2 \\
 $r_{221}$ & 2 & 2 & 1
\end{tabular}
&
\begin{tabular}{@{}c@{~}||@{~}c@{~}@{~}c@{~}@{~}c@{~}@{~}}
 $S$ & $A_S$ & $B_S$ & $D$ \\
\hline
 $s_{111}$ & 1 & 1 & 1 \\
 $s_{112}$ & 1 & 1 & 2 \\
 $s_{121}$ & 1 & 2 & 1 \\
 $s_{211}$ & 2 & 1 & 1
\end{tabular}
&
\begin{tabular}{@{}c@{~}||@{~}c@{~}@{~}c@{~}@{~}}
 $T$ & $A_T$ & $E_T$ \\
\hline
 $t_{12}$ & 1 & 2 \\
 $t_{21}$ & 2 & 1 \\
 $t_{22}$ & 2 & 2
\end{tabular}
&
\begin{tabular}{@{}c@{~}||@{~}c@{~}@{~}c@{~}@{~}}
 $U$ & $E_U$ & $F$ \\
\hline
 $u_{11}$ & 1 & 1 \\
 $u_{21}$ & 2 & 1 \\
 $u_{22}$ & 2 & 2
\end{tabular}
\end{tabular}
\end{small}
\end{center}
\vspace*{-1em}
\caption{Database used in Example~\ref{ex:f-trees2}.}
\label{fig:db2}
\end{figure}

\begin{example}\label{ex:f-trees2}
Consider the query $Q$ from Example~\ref{ex:f-trees} and the f-trees from Figure~\ref{fig:f-trees}. For any database, the left f-tree yields
\[\Phi(\T_1)=\textstyle\sum_A\big[\sum_B\big(\sum_C R \sum_D S\big)\sum_E \big(T \sum_F U\big)\big],\]
while the right f-tree yields
\[\Phi(\T_2)=\textstyle\sum_E\big(\sum_A \big(\sum_B\big(\sum_C R\sum_D S\big)T\big)\big(\sum_F U\big),\]
both in abbreviated form. A procedure to produce the explicit form of $\Phi(\T_1)$ is shown in Figure~\ref{fig:program-ftree}.

For the particular database $\D$ given in Figure \ref{fig:db2}, the f-representations $\Phi(\T_1)$ and $\Phi(\T_2)$ yield the polynomials
\begin{align*}
 P_1 = &(r_{111}(s_{111}+s_{112}) + r_{122}s_{121})t_{12}(u_{21}+u_{22}) + r_{212}s_{211}(t_{21}u_{11} + t_{22}(u_{21} + u_{22})), \\
 P_2 = &r_{212}s_{211}t_{21}u_{11} + ((r_{111}(s_{111}+s_{112}) + r_{122}s_{121})t_{12} + r_{212}s_{211}t_{22})(u_{21} + u_{22}).
\end{align*} 
They are equivalent to each other and to the polynomial $P$ of the flat f-representation of $Q(\mathbf{D})$,
\begin{align*}
 P =& r_{111}s_{111}t_{12}u_{21}+ r_{111}s_{111}t_{12}u_{22}+ r_{111}s_{112}t_{12}u_{21}+\\
     &r_{111}s_{112}t_{12}u_{22}+ r_{122}s_{121}t_{12}u_{21}+ r_{122}s_{121}t_{12}u_{22}+\\
     &r_{212}s_{211}t_{21}u_{11}+ r_{212}s_{211}t_{22}u_{21}+ r_{212}s_{211}t_{22}u_{22}.
\end{align*} 
Whereas $P$ is read-$6$, both $P_1$ and $P_2$ are read-$2$.\punto
\end{example}

\begin{figure}[t!]
\framebox[\columnwidth]{
\parbox{8.7cm}{
\begin{small}
\begin{align*}
&\mbox{{\bf foreach} value $a\in$ Dom$_A$ {\bf do} output sum of}\\
&\hspace*{1em} \mbox{{\bf foreach} value $b\in$ Dom$_B$ {\bf do} output sum of}\\
&\hspace*{2em} \mbox{{\bf foreach} value $c\in$ Dom$_C$ {\bf do output} sum of identifiers of $R$-tuples $(a,b,c)$}\\
&\hspace*{2em} \times\\
&\hspace*{2em}  \mbox{{\bf foreach} value $d\in$ Dom$_D$ {\bf do output} sum of identifiers of $S$-tuples $(a,b,d)$}\hfill\\
&\hspace*{1em} \times\\
&\hspace*{1em} \mbox{{\bf foreach} value $e\in$ Dom$_E$ {\bf do} output sum of}\\
&\hspace*{2em}  \mbox{{\bf output}  sum of identifiers of $T$-tuples $(a,e)$}\\
&\hspace*{2em} \times\\
&\hspace*{2em}  \mbox{{\bf foreach} value $f\in$ Dom$_F$ {\bf do output} sum of identifiers of $U$-tuples $(e,f)$}
\end{align*}
\end{small}
}}
\caption{A procedure for producing $\T_1$-factorisations in explicit form. The abbreviated form is {\small $\textstyle\sum_A\big[\sum_B\big(\sum_C R \sum_D S\big)\sum_E \big(T \sum_F U\big)\big]$}. $\T_1$ is the left f-tree in Figure~\ref{fig:f-trees}.}
\label{fig:program-ftree}
\end{figure}

\begin{remark}\em
For any query $Q$, consider the f-tree $\T$ in which the nodes
labelled by the attribute classes all lie on a single path, and the
leaves labelled by the relations are all attached to the lowest node
in that path. Such a tree $\T$ produces the $\T$-factorisation in
which we sum over all values of all attributes and for each
combination of values we output the product over all relations of the
sums of tuples which have the given values. If all the tuples in the
input relations are distinct, the $\T$-factorisation is just a sum of
products, that is, the flat f-representation of the result.

Thus, for a non-branching tree $\T$ we obtain a flat representation of
$Q(\D)$. The more branching the tree $\T$ has, the more factorised the
$\T$-factorisation of $Q(\D)$ is.\punto
\end{remark}

The correctness of our construction for a general query $Q$ and
database $\D$ is established by the following result.

\begin{proposition}\label{prop:correctness-ftree}
For any f-tree $\T$ of a query $Q$ and any database $\D$, $\Phi(\T)$
is an f-representation of $Q(\D)$.
\end{proposition}

We next introduce definitions concerning f-trees for later
use. Consider an f-tree ${\cal T}$ of a query $Q$. An inner node $A^*$
of ${\cal T}$ is \emph{relevant to a relation $R$} if it contains an
attribute of $R$.  For a relation $R$, let $\mbox{Path}(R)$ be the set
of inner nodes appearing on the path from the leaf $R$ to its root in
${\cal T}$, $\mbox{Relevant}(R)\subseteq \mbox{Path}(R)$ be the set of
nodes relevant to $R$, and $\mbox{Non-relevant}(R) = \mbox{Path}(R)
\setminus \mbox{Relevant}(R)$. For example, in the left f-tree of
Figure~\ref{fig:f-trees}, $\mbox{Non-relevant}(R)=\emptyset$ and
$\mbox{Non-relevant}(U)=\{A^*_R\}$. In the right f-tree,
$\mbox{Non-relevant}(U)$ $= \emptyset$, yet $\mbox{Non-relevant}(R) =
\mbox{Non-relevant}(S)=\{E^*_T\}$.  In fact, there is no f-tree for
the query in Example~\ref{ex:f-trees} such that
$\mbox{Non-relevant}(R)=\emptyset$ for each relation $R$.  This is
because the query is not hierarchical.

\begin{proposition}\label{prop:ftree-hierarchical}
A query is hierarchical iff it has an f-tree ${\cal T}$ such that $\mbox{\em Non-relevant}(R)=\emptyset$ for each relation $R$.
\end{proposition}

The left two trees shown in Figure~\ref{fig:r-f-trees} are f-trees of
a hierarchical query. The first f-tree satisfies the condition in
Proposition~\ref{prop:ftree-hierarchical}, whereas the second does
not.

\section{Readability of Query Results}
\label{sec:query-readability}

The readability of a query $Q$ on a database $\mathbf{D}$ is the
readability of any f-representation of $Q(\D)$, that is, the minimal
possible $k$ such that there exists a read-$k$ representation of
$Q(\D)$. 

In this section we give upper bounds on the readability of arbitrary
select-project-join queries with equality joins in terms of the
cardinality $|\D|$ of the database $\bf{D}$. We then show that these
bounds are asymptotically tight with respect to statically chosen
f-trees. By this we mean that for any query $Q$, if we choose an
f-tree $\T$, there exist arbitrarily large database instances $\D$ for
which the $\T$-factorisation of $Q(\D)$ is read-$k$ with $k$
asymptotically close to our upper bound. In the next section we give
algorithms to compute these bounds. We conclude the section with a
dichotomy: In the class of non-repeating queries, hierarchical queries
are the only queries whose readability for any database is 1 and hence
independent of the size of the database.

A key result for all subsequent estimates of readability is the
following lemma that states the exact number of occurrences of any
identifier of a tuple $\tuple{t}$ in the $\T$-factorisation of $Q(\D)$
as a function of the f-tree $\T$, the query $Q =
\pi_{\bar{A}}(\sigma_\phi (R_1\times \dots
\times R_n))$, and the database $\D$. 

Let $R = R_i$ be a relation of $Q$, denote by the condition ${\cal
S}(R)=\tuple{t}$ the conjunction of equalities of the attributes of
$R$ to corresponding values in $\tuple{t}$, and denote $NR =
\mbox{Non-relevant}(R)$. In the $\T$-factorisation of $Q(\D)$, 
multiple occurrences of the same identifier from $R$ arise from the
summations over the values of attributes from
$NR$. Lemma~\ref{lemma:number-occurrences} quantifies how many
different choices of such values in the summations thus yield a given
identifier from $R$. Recall that the projection attributes $\bar{A}$
do not influence the cardinality of the query result and hence the
number of occurrences of its identifiers, since we consider bag
semantics.

\begin{lemma}\label{lemma:number-occurrences}
The number of occurrences of the identifier $r$ of a tuple $\tuple{t}$
from $R$ in the $\T$-factorisation of $Q(\D)$ is
\[ \left|\left| \big(\pi_{NR} ( \sigma_{{\cal S}(R)=\tuple{t}} \sigma_\phi ( R_1\times \dots \times R_n))\big)(\D) \right|\right|.\]
\end{lemma}

\nop{\begin{lemma}\label{lemma:number-occurrences}
Let $\T$ be an f-tree of a query $Q$. The number of occurrences of the
identifier $r$ of a tuple $\tuple{t}$ from $R$ in the
$\T$-factorisation $\Phi(\T)$ is
\[ \left|\left| \pi_{NR} ( \sigma_{{\cal S}(R)=\tuple{t}} \sigma_\phi ( R_1\times \dots \times R_n)) \right|\right|.\]
\end{lemma}}

For example, for the left f-tree in Figure~\ref{fig:f-trees}, all
identifiers in $R$, $S$, and $T$ occur once, whereas any identifier of
$U$ may occur as many times as distinct $A^*$ values in $R$, $S$, and
$T$. For the leftmost f-tree in Figure~\ref{fig:r-f-trees}, all
identifiers in all relations occur once, since no relation has
non-relevant nodes.

Lemma~\ref{lemma:number-occurrences} represents an effective tool to
further estimate the readability and size of $\T$-factorisations. Our
results build upon existing bounds for query result sizes and yield
readability bounds which can be inferred statically from the
query. Lemma~\ref{lemma:number-occurrences} can be potentially also
coupled with estimates on selectivities and various assumptions on
attribute-value
correlations~\cite{Muralikrishna:SIGMOD:1988,Poosala:VLDB:1997,Getoor:SIGMOD:2001,Re:Cardinality:2010}
to infer database-specific estimates on the readability.

\subsection{Upper Bounds}

Let $\mathbf{D}$ be a database, let $Q =
\pi_{\bar{A}}(\sigma_\phi(R_1\times \dots \times R_n))$ be a query,
let $\mathcal{T}$ be an f-tree of $Q$, and let $R$ be a relation in
$Q$. Denote $NR = \textrm{Non-relevant}(R)$, by $\phi_R$ the condition
$\phi$ restricted to the attributes of $NR$, by $Q_R$ the query
$\sigma_{\phi_R}(\pi_{NR}R_1\times \dots \times \pi_{NR}R_n)$, and by
$\mathbf{D}_R$ the database obtained by projecting each relation in
$\mathbf{D}$ onto the attributes of $NR$.

\begin{lemma}
\label{lem:number-occurrences2}
The number of occurrences of any identifier $r$ from $R$ in the
$\T$-factorisation of $Q(\D)$ is at most
$||Q_R(\mathbf{D}_R)||$.
\end{lemma}
\begin{proof}
By Lemma~\ref{lemma:number-occurrences}, the number of occurrences of
$r$ is equal to
\[\left|\left| \big(\pi_{NR} ( \sigma_{{\cal S}(R)=\tuple{t}} \sigma_\phi ( R_1\times \dots \times R_n))\big)(\D) \right|\right|,\]
from which we obtain the desired bound by straightforward estimates:
\begin{align*}
&|| \big(\pi_{NR} ( \sigma_{{\cal S}(R)=\tuple{t}} \sigma_\phi ( R_1\times
  \dots \times R_n))\big)(\D) || \\
  \leq &|| \big(\pi_{NR} ( \sigma_\phi ( R_1\times \dots \times R_n))
	\big)(\D) || \\
  \leq &|| \big(\sigma_{\phi_R} ( \pi_{NR} ( R_1\times \dots \times R_n))
	\big)(\D) || \\
  = &|| Q_R(\mathbf{D}_R)||. \qedhere
\end{align*}
\end{proof}

The number of distinct tuples in an equi-join query such as $Q_R$ can
be estimated in terms of the database size using the results in
\cite{AGM08}. Intuitively, if we can cover all attributes of the query
$Q_R$ by some $k$ of its relations, then $||Q_R(\D_R)||$ is at most
the product of the sizes of these relations, which is in turn at most
$|\mathbf{D}|^k$. This corresponds to an edge cover of size $k$ in the
hypergraph of $Q_R$. The following result strenghtens this idea by
lifting covers to a weighted version.

\begin{definition}
\label{def:fractional-cover-number}
For an equi-join query $Q = \sigma_\phi(R_1 \times \dots \times R_n)$, the \emph{fractional edge cover number} $\rho^*(Q)$ is the cost of an optimal solution to the linear program with variables $\{x_i\}_{i=1}^n$,
\begin{align*}
\hspace{3em}\textrm{minimising}\qquad&\textstyle\sum_i x_i \\
\textrm{subject to}\qquad&\textstyle\sum_{i : R_i \in r(A)} x_i \geq 1 \qquad\textrm{for all attributes $A$, and}\\
& x_i \geq 0 \qquad\qquad\qquad\:\:\:\,\textrm{for all $i$.\hspace{9em}\punto}
\end{align*} 
\end{definition}

\begin{lemma}[\cite{AGM08}]
\label{lem:fcover-upper}
For any equi-join query $Q$ and for any database $\D$, we have
$||Q(\D)|| \leq |\D|^{\rho^*(Q)}$.
\end{lemma}

Together with Lemma~\ref{lem:number-occurrences2}, this yields the
following bound.

\begin{corollary}
\label{cor:upper}
The number of occurrences of any identifier $r$ from $R$ in the
$\T$-factorisation of $Q(\D)$ is at most $|\D|^{\rho^*(Q_R)}$.
\end{corollary}
\begin{proof}
By Lemma~\ref{lem:number-occurrences2}, the number of occurrences of
$r$ in the $\T$-factorisation of $Q(\D)$ is bounded above by
$||Q_R(\D_R)||$. By Lemma~\ref{lem:fcover-upper}, this is bounded
above by $|\D_R|^{\rho^*(Q_R)}$, which is equal to
$|\D|^{\rho^*(Q_R)}$.
\end{proof}

Corollary \ref{cor:upper} gives an upper bound on the number of
occurrences of identifiers from each relation. Let $M$ be the maximal
number of relations which can contain the same identifier, that is,
the maximal number of relations in $Q$ mapping to the same relation
name by $\mu$. Defining $f(\mathcal{T}) = \max_R \rho^*(Q_R)$ to be
the maximal possible $\rho^*(Q_R)$ over all relations $R$ from $Q$, we
obtain an upper bound on the readability of the $\T$-factorisation of
$Q(\D)$.

\begin{corollary}
The $\T$-factorisation of $Q(\D)$ is at most
\mbox{read-$(M\cdot|\D|^{f(\T)})$}.
\end{corollary}

By considering the $\T$-factorisation with lowest readability, we obtain an upper bound on the readability of $Q(\D)$. Let $f(Q) = \min_\T f(\T)$ be the minimal possible $f(\T)$ over all f-trees $\T$ for $Q$.

\begin{corollary}
\label{cor:upperbound}
For any query $Q$ and any database $\D$, the readability of $Q(\D)$ is at most $M\cdot|\D|^{f(Q)}$.
\end{corollary}

Since $M \leq |Q|$, the readability of $Q(\D)$ is at most $|Q| \cdot
|\D|^{f(Q)}$.

\begin{example}\label{ex:upper-bound}
For the query $Q$ in Example~\ref{ex:f-trees} and the left f-tree in Figure~\ref{fig:f-trees}, the relation $U$ is the only one with a non-empty query $Q_U = \sigma_{\phi_U} (\pi_{A_R} R \times \pi_{A_S} S \times \pi_{A_T} T)$, where the condition $\phi_U$ is $A_R = A_S = A_T$. Since the other relations have empty covers (thus of cost zero), we conclude that their identifiers occur at most once in the query result. We can cover $Q_U$ with any subset of $R$, $S$, and $T$. A minimal edge cover can be any of the relations, and the number of occurrences of any identifier of $U$ is thus linear in the size of that relation. The fractional edge cover number is also 1 and we obtain the same bound.

For the right f-tree in Figure~\ref{fig:f-trees}, both $R$ and $S$ have non-empty queries $Q_R$ and $Q_S$ defining their non-relevant sub-query of $Q$: $Q_R = Q_S = \sigma_{\phi} (\pi_{E_T} T \times \pi_{E_U} U)$, where $\phi$ is $E_T = E_U$.  The attributes $E_T$ and $E_U$ can be covered by $U$ or by $T$. A minimal cover thus has size 1. The minimal fractional edge cover has also cost 1.

Now consider a different query over the relations $R(A_R,E_R)$,
$S(A_S,B_S,C_S)$, $T(A_T,B_T,D_T)$ and $U(C_U,D_U,E_U)$, given by
$\hat{Q} = \sigma_\phi (R\times S \times T \times U)$, with $\phi =
(A_R = A_S = A_T, B_S = B_T, C_S = C_U, D_T = D_U, E_R = E_U)$.

\begin{figure}[t]
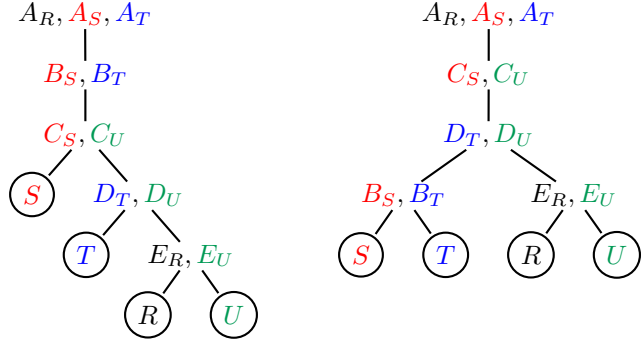

\begin{small}
\[
\psset{levelsep=8mm, nodesep=1pt, treesep=5mm}
\pstree{\TR{A_R,{\color{red}A_S},{\color{blue}A_T}}}
{
  \pstree{\TR{{\color{red}B_S},{\color{blue}B_T}}}
  {
    \pstree{\TR{{\color{red}C_S},{\color{ForestGreen}C_U}}}
    {
      \Tcircle{{\color{red}S}}
      \pstree{\TR{{\color{blue}D_T},{\color{ForestGreen}D_U}}}
      {
        \Tcircle{{\color{blue}T}}
        \pstree{\TR{E_R,{\color{ForestGreen}E_U}}}
        {
          \Tcircle{R}
          \Tcircle{{\color{ForestGreen}U}}
        }
      }
    }
  }
}%
\hspace*{3em}%
\psset{levelsep=8mm, nodesep=1pt, treesep=5mm}
\pstree{\TR{A_R,{\color{red}A_S},{\color{blue}A_T}}}
{
  \pstree{\TR{{\color{red}C_S},{\color{ForestGreen}C_U}}}
  {
    \pstree{\TR{{\color{blue}D_T},{\color{ForestGreen}D_U}}}
    {
      \pstree{\TR{{\color{red}B_S},{\color{blue}B_T}}}
      {
        \Tcircle{{\color{red}S}}
        \Tcircle{{\color{blue}T}}
      }
      \pstree{\TR{E_R,{\color{ForestGreen}E_U}}}
      {
        \Tcircle{R}
        \Tcircle{{\color{ForestGreen}U}}      
      }
    }
  }
}%
\]
\end{small}
\caption{F-trees $\T_1$ and $\T_2$ for the query in Example~\ref{ex:upper-bound}.}
\label{fig:f-tree-upper}
\end{figure}

Consider the left f-tree $\T_1$ shown in Figure~\ref{fig:f-tree-upper}. For the relation $R$, we have $\mbox{Non-relevant}(R) = \{B_S^*,C_S^*,D_T^*\}$, and hence the restricted query $Q_R$ will be $\hat{Q}_R = \sigma_{B_S=B_T,C_S=S_U,D_T=D_U}(\pi_{B_S,C_S}S \times \pi_{B_T,D_T} T \times \pi_{C_U,D_U}U)$. We need at least two of the relations $S,T,U$ to cover all attributes of $Q_R$, the edge cover number is thus 2. However, in the fractional edge cover linear program, we can assign to each relation the value $x_S = x_T = x_U = 1/2$. The covering conditions at each attribute are satisfied, since each attribute belongs to two of the relations. The total cost of this solution is only $3/2$. It is in fact the optimal solution, so $\rho^*(Q_R) = 3/2$. It is easily seen that $\rho^*(\hat{Q}_T) = \rho^*(\hat{Q}_U) = 1$ (since $\hat{Q}_T$ can be covered either by $S$ or $U$, and $\hat{Q}_U$ can be covered by either $S$ or $T$) and $\rho^*(\hat{Q}_S) = 0$ (since $\hat{Q}_S$ has no attributes), so $f(\T_1) = 3/2$. We obtain the upper bound $|\D|^{3/2}$ on the number of occurrences of identifiers from $R$, and hence on the readability of any $\T_1$-factorisation.

Note however that in the right f-tree $\T_2$ in Figure~\ref{fig:f-tree-upper}, each of $\hat{Q}_R$, $\hat{Q}_S$, $\hat{Q}_T$ and $\hat{Q}_U$ is covered by only one of its relations, and hence $f(\T_2) = 1$. Any $\T_2$-factorisation will therefore have readability at most linear in $\D$.

In fact, no f-tree $\T$ for $\hat{Q}$ has $f(\T) < 1$, so $\T_2$ is in this sense optimal and $f(\hat{Q}) = 1$.
\punto
\end{example}

\subsection{Lower Bounds}

We also show that the obtained bounds on the numbers of occurrences of identifiers are essentially tight. For any query $Q$ and any f-tree $\T$, we construct arbitrarily large databases for which the number of occurrences of some symbol is asymptotically as large as the upper bound.

The expression for the number of occurrences of an identifier, given
in Lemma~\ref{lemma:number-occurrences}, states the size of a specific
query result. As a first attempt to construct a small database $\D$
with a large result for the query $Q_R$, we pick $k$ attribute classes
of $Q_R$ and let each of them attain $N$ different values. If each
relation has attributes from at most one of these classes, each
relation in $\D$ will have size at most $N$, while the result of $Q_R$
will have size $N^k$.

This corresponds to an independent set of $k$ nodes in the hypergraph
of $Q_R$. We can again strenghten this result by lifting independent
sets to a weighted version. Since the edge cover and the independent
set problems are dual when written as linear programming problems,
this lower bound meets the upper bound from the previous
subsection. The following result, derived from results
in~\cite{AGM08}, forms the basis of our argument.

\begin{lemma}
\label{lem:fcover-lower}
For any equi-join query $Q$, there exist arbitrarily large databases $\D$ such that $||Q(\D)|| \geq (|\D|/|Q|)^{\rho^*(Q)}$.
\end{lemma}

Now let $Q = \pi_{\bar{A}}(\sigma_\phi (R_1\times \dots \times R_n))$
be a query, let $\mathcal{T}$ be an f-tree of $Q$ and let $R$ be a
relation in $Q$. Define $NR$, $\phi_R$ and $Q_R$ as before. We can
apply Lemma~\ref{lem:fcover-lower} to the expression from
Lemma~\ref{lemma:number-occurrences} to infer lower bounds for numbers
of occurrences of identifiers in the $\T$-factorisation of $Q(\D)$.

\begin{lemma}
\label{lem:lower}
There exist arbitrarily large databases $\D$ such that each identifier
from $R$ occurs in the $\T$-factorisation of $Q(\D)$ at least
$(|\D|/|Q|)^{\rho^*(Q_R)}$ times.
\end{lemma}

We now lift the result of Lemma~\ref{lem:lower} from the identifiers
from $R$ to all identifiers in the $\T$-factorisation of $Q(\D)$.

\begin{corollary}
There exist arbitrarily large databases $\D$ such that the
$\T$-factorisation of $Q(\D)$ is at least read-$(|\D|/|Q|)^{f(\T)}$.
\end{corollary}

Finally, by minimising over all f-trees $\T$, we find a lower bound on
readability with respect to statically chosen f-trees.

\begin{corollary}
\label{cor:lowerbound}
Let $Q$ be a query. For any f-tree $\T$ of $Q$ there exist arbitrarily
large databases $\D$ for which the $\T$-factorisation of $Q(\D)$ is at
least read-$(|\D|/|Q|)^{f(Q)}$.
\end{corollary}

\begin{example}
Let us continue Example~\ref{ex:upper-bound}. For the left f-tree in Figure~\ref{fig:f-trees}, an independent set of attributes covering the relations $R$, $S$, and $T$ of the query $Q_U$ is $\{A_R^*\}$. Since $Q_U$ only has one attribute, this is also the largest independent set, and the fractional relaxation of the maximum independent set problem has also optimal cost 1.

For the right f-tree in Figure~\ref{fig:f-trees} the situation is similar. A maximum independent set of attributes covering the relations $T$ and $U$ of the queries $Q_R$ and $Q_S$ is $\{E_T^*\}$ and has size 1.

The situation is more interesting for the query $\hat{Q}$. Recall that for the left f-tree $\T_1$ in Figure~\ref{fig:f-tree-upper}, $\hat{Q}_R = \sigma_{B_S=B_T,C_S=S_U,D_T=D_U}(\pi_{B_S,C_S}S \times \pi_{B_T,D_T} T \times \pi_{C_U,D_U}U)$, its attribute classes being $NR = \{B_S^*,C_S^*,D_T^*\}$. The maximum independent set for $\hat{Q}_R$ has size 1, since any two of its attribute classes are relevant to a common relation. However, the fractional relaxation of the maximum independent set problem allows to increase the optimal cost to $3/2$. In this relaxation, we want to assign nonnegative rational values to the attribute classes, so that the sum of values in each relation is at most one. By assigning to each attribute class the value $1/2$, the sum of values in each relation is equal to one, and the total cost of this solution is $3/2$. This is used in the proof of Lemma~\ref{lem:lower} to construct databases $\D$ for which the identifiers from $R$ appear at least $(|\D|/3)^{3/2}$ times in the $\T_1$-factorisation of $\hat{Q}(\D)$, thus proving the upper bound from Example~\ref{ex:upper-bound} asymptotically tight.

Since all f-trees $\T$ for $\hat{Q}$ have $f(\T) \geq 1$, the results
in this subsection show that for any such f-tree $\T$ we can find
databases $\D$ for which the readability of the $\T$-factorisation of
$Q(\D)$ is at least linear in $|\D|$.
\punto
\end{example}

\subsection{Characterisation of Queries by Readability}

For a fixed query, the obtained upper and lower bounds meet
asymptotically. Thus our parameter $f(Q)$ completely characterises
queries by their readability with respect to statically chosen
f-trees.

\begin{theorem}
\label{th:characterisation}
Fix a query $Q$. For any database $\D$, the readability of $Q(\D)$ is
$O(|\D|^{f(Q)})$, while for any f-tree $\T$ of $Q$, there exist
arbitrarily large databases $\D$ for which the $\T$-factorisation of
$Q(\D)$ is read-$\Theta(|\mathbf{D}|^{f(Q)})$.
\end{theorem}

Theorem~\ref{th:characterisation} subsumes the case of hierarchical queries.

\begin{corollary}
\label{cor:hierar}
Fix a query $Q$. If $Q$ is hierarchical, the readability of $Q(\D)$
for any database $\D$ is bounded by a constant. If $Q$ is
non-hierarchical, for any f-tree $\T$ of $Q$ there exist arbitrarily large databases
$\D$ such that the $\T$-factorisation of $Q(\D)$ is
read-$\Theta(|\D|)$.
\end{corollary}

For non-repeating queries, the following result extends the above
dichotomy to the case of readability irrespective of f-trees.

\begin{theorem}
\label{th:unbounded}
Fix a non-repeating query $Q$. If $Q$ is hierarchical, then the
readability of $Q(\D)$ is 1 for any database $\D$. If $Q$ is
non-hierarchical, then there exist arbitrarily large databases $\D$
such that the readability of $Q(\D)$ is $\Omega(\sqrt{|\D|})$.
\end{theorem}

\begin{figure}[t]
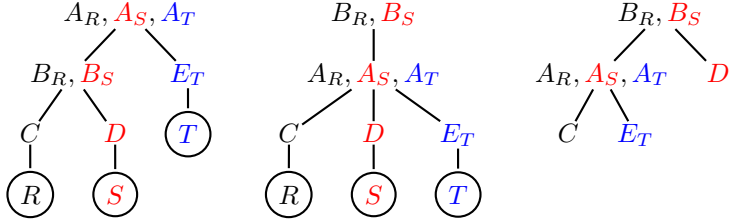

\begin{small}
\[
\psset{levelsep=8mm, nodesep=1pt, treesep=5mm}
\pstree{\TR{A_R,{\color{red}A_S},{\color{blue}A_T}}}
{
  \pstree{\TR{B_R,{\color{red}B_S}}}
  {
   \pstree{\TR{C}}{\Tcircle{R}}
   \pstree{{\color{red}\TR{D}}}{{\color{red}\Tcircle{S}}}
  }
  \pstree{\TR{{\color{blue}E_T}}}{{\color{blue}\Tcircle{T}}}
}%
\hspace*{2em}%
\psset{levelsep=8mm, nodesep=1pt, treesep=5mm}
\pstree{\TR{B_R,{\color{red}B_S}}}
{
  \pstree{\TR{A_R,{\color{red}A_S},{\color{blue}A_T}}}
  {
   \pstree{\TR{C}}{\Tcircle{R}}
   \pstree{{\color{red}\TR{D}}}{{\color{red}\Tcircle{S}}}
   \pstree{\TR{{\color{blue}E_T}}}{{\color{blue}\Tcircle{T}}}
  }
}%
\hspace*{2em}%
\psset{levelsep=8mm, nodesep=1pt, treesep=5mm}
\pstree{\TR{B_R,{\color{red}B_S}}}
{
  \pstree{\TR{A_R,{\color{red}A_S},{\color{blue}A_T}}}
  {
   \TR{C}
   \TR{{\color{blue}E_T}}
  }
  {\color{red}\TR{D}}
}%
\]
\end{small}

\caption{Left to right: Two f-trees and a tree which cannot be extended to an f-tree, used in Example~\ref{ex:r-f-trees}.}
\label{fig:r-f-trees}
\end{figure}

%%%%%%%%%%%%%%%%%%%%%%%%%%%%%%%%%%%%%%%%%%%%%%%
\section{Algorithms for Query Characterisation}
\label{sec:algorithm}

Given a query $Q$, we show how to compute the parameter $f(Q)$
characterising the upper bound on readability. We give an algorithm
that iterates over all f-trees $\T$ of $Q$ to find one with minimum
$f(\T)$. We further prune the space of possible f-trees to avoid
suboptimal choices.

The following lemma facilitates the search for optimal f-trees. Intuitively, since the parameter $f(\T)$ depends on the costs of fractional covers of $Q_R$ for the relations $R$ of $Q$, and since $Q_R$ is the restriction of $Q$ to the attributs of $\mbox{Non-relevant}(R) = \mbox{Path}(R) \setminus \mbox{Relevant}(R)$, by shrinking the sets $\mbox{Path}(R)$, the fractional cover number of $Q_R$ and hence the parameter $f(\T)$ can only decrease.

\begin{lemma}
\label{lem:dec-path}
If $\T_1$ and $\T_2$ are f-trees for a query $Q$, and
$\textrm{Path}(R)$ in $\T_1$ is a subset of $\textrm{Path}(R)$ in $\T_2$ for any
relation $R$ of $Q$, then $f(\T_1) \leq f(\T_2)$.
\end{lemma}

In any f-tree $\T$, each relation symbol $R$ lies under its lowest
relevant node $A^*$. By moving $R$ upwards directly under $A^*$,
$\textrm{Path}(R)$ can only shrink, and by Lemma~\ref{lem:dec-path},
$f(\mathcal{T})$ can only decrease. Thus, when iterating over all
possible f-trees $\mathcal{T}$ to find one with lowest
$f(\mathcal{T})$, we can assume that the leaves are as close as
possible to the root, and it is enough to iterate over all the
possible subtrees formed by the inner nodes of f-trees. We next denote by 
\emph{reduced f-trees} the f-trees where the leaves are removed.
The only condition for a rooted tree over the set of nodes labelled by
the attribute classes of $Q$ to be a reduced f-tree, is that for each
relation $R$, no two nodes relevant to $R$ lie in sibling
subtrees. Call this condition $\mathcal{C}$.

\begin{example}
\label{ex:r-f-trees}
Consider the relations $R$, $S$ and $T$ over sche\-mas
$\{A_R,B_R,C\}$, $\{A_S,B_S,D\}$ and $\{A_T,E_T\}$ respectively, and
the query $Q = \sigma_{\phi}(R\times S\times T)$ with $\phi =
(A_R=A_S,A_R=A_T,B_R=B_S)$. Figure~\ref{fig:r-f-trees} depicts three
trees. Without their leaves, the first two are reduced f-trees.  The
third tree is not a reduced f-tree as it violates condition
$\mathcal{C}$: the nodes $A_S^*$ and $D^*$ lie in sibling subtrees,
yet they are both relevant to $S$. We cannot place the leaf $S$ under
both of them.\punto
\end{example}

Any reduced f-tree is a rooted forest satisfying the condition
$\mathcal{C}$. Such a forest can either be a single rooted tree, or a
collection of rooted trees. In the first case, the condition
$\mathcal{C}$ on the whole tree rooted at $A^*$ is equivalent to
$\mathcal{C}$ on the collection of subtrees of $A^*$. In the second
case, the condition $\mathcal{C}$ must hold in the individual
subtrees, but in addition, for each relation $R$, the set of its
relevant nodes $\textrm{Relevant}(R)$ can only intersect one of the
subtrees. This recursive characterisation of the condition
$\mathcal{C}$ is used in the {\bf iter} algorithm in
Figure~\ref{fig:algo1} to enumerate all reduced f-trees of a query
with the set $S$ of attribute classes.

\begin{figure}[t!]
\framebox[\columnwidth]{
\parbox{12cm}{
\begin{small}
\medskip
Call a partition $P_1,\dots,P_n$ \emph{good} if for each relation $R$ in $Q$, the nodes relevant to $R$ lie in at most one $P_i$.
\begin{align*}
&\hspace*{-1em} \underline{\mbox{\bf{}iter}} \mbox{(node set $S$)}\\
&\mbox{\bf foreach } A^* \in S \mbox{\bf{} do}\tag{1}\\
&\hspace*{2em} \mbox{\bf foreach } \mathcal{T} \in \mbox{{\bf{}iter}($S \setminus \{A^*\}$)} \mbox{\bf{} do}\\
&\hspace*{4em} \mbox{\bf output} \mbox{ tree formed by root $A^*$ and child $\mathcal{T}$} \\
&\mbox{{\bf foreach} good partition } P_1,\dots,P_n \mbox{ of } S \mbox{\bf{} do}\tag{2}\\
&\hspace*{2em}
\mbox{\bf foreach } (\mathcal{T}_1,\dots,\mathcal{T}_n) \in (\mbox{\bf iter}(P_1),\dots,\mbox{\bf iter}(P_n))  \mbox{\bf\ do}\\
&\hspace*{4em}
\mbox{\bf output } \mathcal{T}_1\cup\cdots\cup \mathcal{T}_n
\end{align*}
\end{small}\vspace*{-1em}
 \centering
}}
\caption{Iterating over all reduced f-trees.}
\label{fig:algo1}
\end{figure}

\begin{example}
Consider the query in Example~\ref{ex:r-f-trees}. When algorithm
\textbf{iter} chooses the root $\{A_R,A_S,A_T\}$ in step 1, in the
next recursive call it can split the remaining notes into $P_1 =
\{B_R^*,C^*,D^*\}$ and $P_2 = \{E_T^*\}$, since $\textrm{Relevant}(R)$
and $\textrm{Relevant}(S)$ only intersect $P_1$ and
$\mathrm{Relevant}(T)$ only intersects $P_2$. The first tree in
Figure~\ref{fig:r-f-trees} is created like this. However, when we
choose $\{B_R,B_S\}$ in step 1, in the next recursive call there are
no possible partitions in step 2, since the node $\{A_R,A_S,A_T\}$
lies in all of $\textrm{Relevant}(R)$, $\textrm{Relevant}(S)$,
$\textrm{Relevant}(T)$. The second tree in Figure~\ref{fig:r-f-trees}
is created within this call, while the third tree in
Figure~\ref{fig:r-f-trees}, which is not a valid reduced f-tree, is
never produced.\punto
\end{example}

However, some choices of the root in line (1) and some choices of
partitioning in line (2) of {\bf iter} are suboptimal. Firstly we have

\begin{lemma}
\label{lem:swap}
Let $\T$ be an f-tree. For two nodes $A^*$ and $B^*$, if $r(B) \subset r(A)$ and $B^*$ is an ancestor of $A^*$, then by swapping them we do not violate the
condition ${\cal C}$ and do not increase $f(\mathcal{T})$.
\end{lemma}

\noindent Thus, we do not need to consider trees with root $B^*$. The second
tree in Figure~\ref{fig:r-f-trees} is suboptimal, since $B^*$ is the
root instead of $A^*$. If $r(B) = r(A)$, then $A^*$ and $B^*$ are
interchangeable in any f-tree, and we need only consider one of them
as the root.

Secondly, in line (2) of {\bf iter}, among all the good partitions, there always exists a finest one. That is, there always exists a finest partition $P_1,\dots,P_n$ of the attribute classes such that $\textrm{Relevant}(R)$ only intersects
one $P_i$ for each relation $R$. We do not need to consider any coarser partitions in line (2): for any such coarser partition, we could split one of its trees into two, while not increasing $\textrm{Path}(R)$ for any $R$ and thus not increasing $f(\T)$ by Lemma~\ref{lem:dec-path}. Moreover, if $n>1$, by a similar argument we do not need to execute line (1) at all, increasing the fanout of a node is always better. These observations lead to a pruned version of algorithm \textbf{iter}, given in Figure~\ref{fig:algo2}.

\begin{figure}[t!]
\framebox[\columnwidth]{
\parbox{12cm}{
\begin{small}
\medskip
Define a partial order on the nodes of $Q$ by $A^* > B^*$ iff either $r(A) \supset r(B)$, or $r(A)=r(B)$ and $A^*$ is lexicographically larger than $B^*$ (to break ties arbitrarily among interchangeable nodes). Also, call a partition $P_1,\dots,P_n$ \emph{good} if for each relation $R$ in $Q$, the nodes relevant to $R$ lie in at most one $P_i$.
\begin{align*}
&\hspace*{-1em} \underline{\mbox{\bf{}iter-pruned}} \mbox{(node set $S$)}\\
&\mbox{\bf let }P_1,\dots,P_n\mbox{ be the finest good partition of $S$}\\
&\mbox{\bf if }n=1\mbox{\bf{} then}\\
&\hspace*{2em}\mbox{\bf foreach } \mbox{$>$-maximal } A^* \in S \mbox{\bf{} do}\\
&\hspace*{4em} \mbox{\bf foreach } \mathcal{T} \in \mbox{{\bf{}iter-pruned}($S \setminus \{A^*\}$)} \mbox{\bf{} do}\\
&\hspace*{6em} \mbox{\bf output} \mbox{ tree formed by root $A^*$ with $\mathcal{T}$ as its child} \\
&\mbox{\bf else}\\
&\hspace*{2em}\mbox{\bf foreach }(\mathcal{T}_1,\dots,\mathcal{T}_n) \in (\mbox{\bf{}iter-pruned}(P_1),\dots,\mbox{\bf{}iter-pruned}(P_n)) \mbox{\bf{} do} \\
&\hspace*{4em} \mbox{\bf output} \mbox{ $\mathcal{T}_1 \cup \dots \cup \mathcal{T}_n$}
\end{align*}
 \centering
\end{small}\vspace*{-1em}
}}
\caption{Pruned algorithm \textbf{iter-pruned}.}
\label{fig:algo2}
\end{figure}

For the query in Example~\ref{ex:r-f-trees}, algorithm
\textbf{iter-pruned} does not output the second tree in
Figure~\ref{fig:r-f-trees}. The node $\{B_R,B_S\}$ is not considered
for the root since $r(B_R) \subset r(A_R)$. In fact
\textbf{iter-pruned} only produces the first tree from
Figure~\ref{fig:r-f-trees}, and exhibits such behaviour for all
hierarchical queries:

\begin{proposition}
\label{prop:algo-hierarchical}
For a hierarchical query $Q$, the algorithm \textbf{iter-pruned} has
exactly one choice at each recursive call, and outputs a single
reduced f-tree in polynomial time.
\end{proposition}

Using lazy evaluation, at any moment there are at most linearly many
calls of \textbf{iter} or \textbf{iter-pruned} on the stack. Between
two consecutive output trees, there are at most linearly many
recursive calls. The following theorem summarises our results so far.

\begin{theorem}
Given a query $Q$, the algorithms {\bf iter} and {\bf iter-pruned}
enumerate reduced f-trees of $Q$ with polynomial delay and polynomial
space. Algorithm {\bf iter} enumerates all reduced f-trees, while {\bf
iter-pruned} only a subset of these, which contains one with optimal
$f(\T)$.
\end{theorem}

Both algorithms can enumerate exponentially many reduced f-trees.

\nop{
\begin{proposition}
Consider the relations $R_{i,j}$ of signatures $\{A_{\{i,j\}}^i,A_{\{i,j\}}^j\}$ for $1\leq i \leq j \leq n$, let $\phi = \bigwedge_i (A_{\{1,i\}}^i = \dots = A_{\{n,i\}}^i)$, and let $Q = \sigma_\phi \prod_{i,j} R_{\{i,j\}}$. Then \textbf{iter-pruned} produces $(n-1)!$ different reduced f-trees.
\end{proposition}
}

For each constructed reduced f-tree, we can easily add the leaves with
relations under the lowest node from $\mbox{Relevant}(R)$. For each
such f-tree $\mathcal{T}$, we need to compute $f(\mathcal{T})$, the
maximum of $\rho^*(Q_R)$ over all relations $R$ in $Q$, which can be
done in polynomial time, or using the simplex algorithm for linear
programming.

\begin{figure}[t]
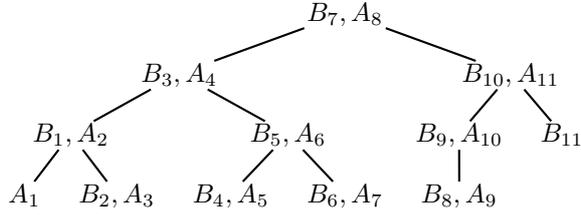

\begin{small}
\[
\psset{levelsep=8mm, nodesep=1pt, treesep=5mm}
\pstree{\TR{B_7,A_8}}
{
  \pstree{\TR{B_3,A_4}}
  {
    \pstree{\TR{B_1,A_2}}
    {
      \TR{A_1}
      \TR{B_2,A_3}
    }
    \pstree{\TR{B_5,A_6}}
    {
      \TR{B_4,A_5}
      \TR{B_6,A_7}
    }
  }
  \pstree{\TR{B_{10},A_{11}}}
  {
    \pstree{\TR{B_9,A_{10}}}
    {
      \TR{B_8,A_9}
    }
    \TR{B_{11}}
  }
}%
\]
\end{small}\vspace*{-1em}

\caption{A reduced f-tree for $Q_{12}$ with $f(\mathcal{T})=2$, the lowest possible.}
\label{fig:chainquery}\vspace*{-1em}
\end{figure}

\begin{example}
  Consider relations $R_i$ over schemas $\{A_i,B_i\}$ for $1\leq i <
  n$, let $\phi = \bigwedge_{i=1}^{n-2} (B_i = A_{i+1})$ and let $Q_n$
  be the query $\sigma_\phi \prod R_i$. This query is a chain of $n-1$
  joins.
  
  A reduced f-tree $\T$ for $Q_{12}$ is shown in Figure~\ref{fig:chainquery}. In the corresponding f-tree, the leaves labelled by all relations apart from $R_{10}$ hang from the leaves of the reduced f-tree. For all relations $R_i$, the query $Q_{R_i}$ has at most two attributes, so $\rho^*(Q_{R_i}) \leq 2$. However, for $R_1$ (and most other $R_i$), each of the four relations of the query $Q_{R_1}$ only has one of the two attributes, so the fractional edge cover number of $Q_{R_1}$ is $2$. It follows that $f(\T) = 2$. In fact, this is the lowest possible value, and $f(Q_{12}) = 2$. For arbitrary $n$, $f(Q_n) = \lfloor \log_2{n} \rfloor - 1$.

  This example shows that branching in f-trees is key to low
  readability. An alternative yet naive approach is to choose a
  minimal set of attributes such that when these attributes are set to
  values from their domain, $Q_n$ becomes hierarchical.  We can then
  sum over all possible domain values for each such attribute, and for
  each combination we create a read-once f-representation. In case of
  $Q_n$, a minimal set of such attributes has cardinality ($\lfloor
  \frac{n}{4} \rfloor$), which is linear in the size of $Q_n$.  The
  corresponding f-tree would have $f(\T) = \Theta(n)$, which is
  exponentially worse than the optimal value.\punto
\end{example}

\section{Algorithms for Computing $\T$-factorisations of Query Results}
\label{sec:falgorithm}

Figure~\ref{fig:algorithm1} gives a high-level recipe for producing
the $\T$-factorisation of $Q(\D)$, given an f-tree $\T$ of a query
$Q$, and a database $\D$. We present here a more detailed
implementation of this algorithm and analyse its performance.

\begin{figure}[t!]
\framebox[\columnwidth]{
\parbox{12cm}{
\begin{small}
\medskip
Let $\mathcal{T}$ be an f-tree for a query $Q$ and let $\mathbf{D}$ be a database. The $\T$-factorisation $Q(\mathbf{D})$ is obtained by running $\mathbf{gen}(\mathcal{T},\top)$, where
\begin{align*}
&\hspace*{-1em} \underline{\mbox{\bf{}gen}} \mbox{(tree $\mathcal{T}$, conjunctive condition $\gamma$)}\\
&\mbox{\bf if }\mbox{$\mathcal{T}$ is a tree with root $A^*$ and children $\mathcal{U}$ }\mbox{\bf then}\\
&\hspace*{2em} \mbox{\bf create }\mbox{f-representation $S$ := an empty sum}\\
&\hspace*{2em} \mbox{\bf foreach }\mbox{value $a$ of any attribute from $A^*$ in the database $\mathbf{D}$} \mbox{\bf{} do}\tag{1}\\
&\hspace*{4em} \mbox{\bf append } \mbox{\bf gen}(\mathcal{U},\gamma \wedge (A^* = a))\mbox{ to $S$} \\
&\hspace*{2em} \mbox{\bf return }S\\
&\mbox{\bf else if }\mbox{$\mathcal{T}$ is a collection of trees $\mathcal{T}_1,\dots,\mathcal{T}_n$ }\mbox{\bf then}\\
&\hspace*{2em} \mbox{\bf return }\mbox{\bf gen}(\mathcal{T}_1,\gamma)\cdot\ldots\cdot\mbox{\bf gen}(\mathcal{T}_n,\gamma) \\
&\mbox{\bf else if }\mbox{$\mathcal{T}$ is a leaf $R$ }\mbox{\bf then}\\
&\hspace*{2em} \mbox{\bf create }\mbox{f-representation $S$ := an empty sum}\\
&\hspace*{2em} \mbox{\bf foreach }\mbox{tuple $\tuple{t_i}$ of $\mathbf{R} = R(\mathbf{D})$ satisfying $\gamma$ } \mbox{\bf{}do}\tag{2}\\
&\hspace*{4em} \mbox{\bf append } r_i\tuple{\pi_{\mathrm{head}(Q)}t_i}\mbox{ to $S$} \\
&\hspace*{2em} \mbox{\bf return }S
\end{align*}
\end{small}\vspace*{-1em}
 \centering
}}
\caption{Naive implementation of the factorisation algorithm.}
\label{fig:falgo1}
\end{figure}

A naive implementation of the factorisation algorithm, exactly mimicking the definition from Figure~\ref{fig:algorithm1}, is given in Figure~\ref{fig:falgo1}. However, it contains two obvious inefficiencies. In line (1), it is inefficient to explicitly iterate over all values $a$ of the attributes from $A^*$, which appear in the database $\mathbf{D}$, because for some of them, $\mbox{\bf gen}(\mathcal{U},\gamma \wedge (A^* = a))$ necessarily produces an empty f-representation. Also, in line (2), it is inefficient to search every time through the entire relation $\mathbf{R}$ for tuples satisfying $\gamma$.

We eliminate these inefficiencies in the implementation $\textbf{gen2}$ given in Figure~\ref{fig:falgo2}, by passing the ranges of tuples satisfying $\gamma$ for each relation of $\D$, instead of the condition $\gamma$. At any call of $\mathbf{gen}$, the set of attributes constrained by $\gamma$ forms a contiguous path ending at a root of the original tree $\mathcal{T}$. Therefore, if we sort the tuples of each relation first by the highest-appearing attribute, then by the next-highest-one, and so on, then in each call of $\textbf{gen}$, the set of tuples satisfied by $\gamma$ will form a contiguous range in each relation. Thus in $\textbf{gen2}$, we only need to pass the pointers to the beginning and end of each such range.

\begin{figure}[t!]
\framebox[\columnwidth]{
\parbox{12cm}{
\begin{small}
\medskip
Let $\T$ be an f-tree for a query $Q$ and let $\D$ be a database. Order the attributes of each relation by their appearance in $\T$ from higest to lowest, and then sort the tuples of each relation, by higher attributes first. The $\T$-factorisation of $Q(\D)$ is obtained by running $\textbf{gen2}(\mathcal{T},\{(1,|\mathbf{R}_i|)\}_{i=1}^n)$, where
\begin{align*}
&\hspace*{-1em} \underline{\mbox{\bf{}gen2}} \mbox{(tree $\mathcal{T}$, ranges $(\mathrm{start}_i,\mathrm{end}_i)$ for $i = 1,\dots,n$)}\\
&\mbox{\bf if }\mbox{$\mathcal{T}$ is a tree with root $A^*$ and children $\mathcal{U}$ }\mbox{\bf then}\\
&\hspace*{2em} \mbox{\bf create }\mbox{f-representation $S$ := an empty sum}\\
&\hspace*{2em} \mbox{\bf repeat} \\
&\hspace*{4em} \mbox{find the next ranges $(start'_i,end'_i) \subseteq (start_i,end_i)$} \\
&\hspace*{4em} \mbox{of tuples sharing the same value $a$ on $A^*$,} \tag{1}\\ 
&\hspace*{4em} \mbox{\bf append } \mbox{\bf gen2}(\mathcal{U},\{(\mathrm{start}'_i,\mathrm{end}'_i)\}_{i=1}^n)\mbox{ to $S$} \\
&\hspace*{2em} \mbox{\textbf{until} no more such ranges exist} \\
&\hspace*{2em} \mbox{\bf return }S\\
&\mbox{\bf else if }\mbox{$\mathcal{T}$ is a collection of trees $\mathcal{T}_1,\dots,\mathcal{T}_n$ }\mbox{\bf then}\\
&\hspace*{2em} \mbox{\bf return } \textstyle\prod_{j=1}^n \mbox{\bf gen2}(\mathcal{T}_j,\{(\mathrm{start}_i,\mathrm{end}_i)_{i=1}^n\}) \\
&\mbox{\bf else if }\mbox{$\mathcal{T}$ is a leaf $R_i$ }\mbox{\bf then}\\
&\hspace*{2em} \mbox{\bf return } \textstyle\sum_{j=\mathrm{start}_i}^{\mathrm{end}_i} r_j\tuple{\pi_{\mathrm{head}(Q)}t_j} \tag{2}
\end{align*}
\end{small}\vspace*{-1em}
 \centering
}}
\caption{Improved implementation of the factorisation algorithm.}
\label{fig:falgo2}\vspace*{-1em}
\end{figure}

Moreover, if $\mathcal{T}$ in
$\mathbf{gen2}(\mathcal{T},\mathrm{ranges})$ is a tree with root
$A^*$, for each relation $R$ with an attribute $A \in A^*$, the tuples
of $\mathbf{R}$ in the corresponding range will be sorted by the
attribute $A$. When iterating over values $a$ of $A^*$ in line (1),
using a mergesort-like strategy we can find those values $a$ which
appear at least once in the relevant range of \emph{each} relation in
$r(A)$. For each such $a$ we also find the corresponding range of
tuples in each relation and recurse. For other $a$, i.e.\ those for
which at least one relation with an attribute in $A^*$ has no tuples
in the current range with value $a$ in that attribute, the
f-representation generated at $A^*$'s children would be empty, and we
do not need to recurse.

Finally, if $\T$ is just a leaf $R$, the iteration in line (2) becomes
trivial in $\mathbf{gen2}$, we simply iterate over the corresponding
range in $\mathbf{R}$.

\begin{example}
Consider the left f-tree $\mathcal{T}$ of Figure~\ref{fig:f-trees} and
the database $\mathbf{D}$ used in Example~\ref{ex:f-trees2}, also
shown in Figure~\ref{fig:exec}. Let us examine the execution of the
call $\textbf{gen2}(\mathcal{T},\mathcal{R})$, where $\mathcal{R}$
represents the full range in each relation of $\D$. The root of
$\mathcal{T}$ is the node $\{A_R,A_S,A_T\}$, relevant to the relations
$R$, $S$ and $T$. The first execution of line (1) finds the ranges
given in red in Figure~\ref{fig:exec}, with the common value of
$A_R=A_S=A_T=1$. Notice that $U$ does not have an attribute in the
root node, so its range remains unchanged.

\begin{figure}[t]
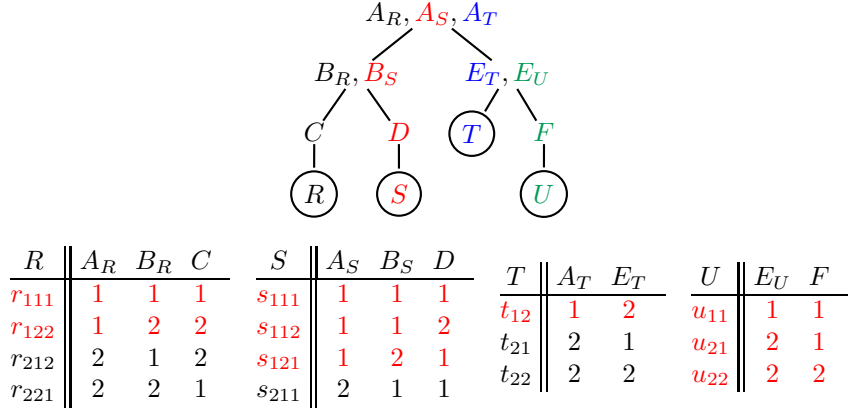

\begin{center}
\begin{small}
\[
\psset{levelsep=8mm, nodesep=1pt, treesep=5mm}
\pstree{\TR{A_R,{\color{red}A_S},{\color{blue}A_T}}}
{
  \pstree{\TR{B_R,{\color{red}B_S}}}
  {
   \pstree{\TR{C}}
    {
      \Tcircle{R}
    }
   \pstree{{\color{red}\TR{D}}}
    {
      \Tcircle{{\color{red}S}}
    }
  }
  \pstree{\TR{{\color{blue}E_T},{\color{ForestGreen}E_U}}}
  {
    \Tcircle{{\color{blue}T}}
    \pstree{{\color{ForestGreen}\TR{F}}}
     {
       \Tcircle{{\color{ForestGreen}U}}
     }
  }
}%
\]
\medskip
\begin{tabular}{c@{\quad}c@{\quad}c@{\quad}c}
\begin{tabular}{@{}c@{~}||@{~}c@{~}@{~}c@{~}@{~}c@{~}@{~}}
 $R$ & $A_R$ & $B_R$ & $C$ \\
\hline
 \color{red} $r_{111}$ & \color{red} 1 & \color{red}1 & \color{red} 1 \\
 \color{red} $r_{122}$ & \color{red} 1 & \color{red}2 & \color{red}2 \\
 $r_{212}$ & 2 & 1 & 2 \\
 $r_{221}$ & 2 & 2 & 1
\end{tabular}
&
\begin{tabular}{@{}c@{~}||@{~}c@{~}@{~}c@{~}@{~}c@{~}@{~}}
 $S$ & $A_S$ & $B_S$ & $D$ \\
\hline
 \color{red}$s_{111}$ & \color{red}1 & \color{red}1 & \color{red}1 \\
 \color{red}$s_{112}$ & \color{red}1 & \color{red}1 & \color{red}2 \\
 \color{red}$s_{121}$ & \color{red}1 & \color{red}2 & \color{red}1 \\
 $s_{211}$ & 2 & 1 & 1
\end{tabular}
&
\begin{tabular}{@{}c@{~}||@{~}c@{~}@{~}c@{~}@{~}}
 $T$ & $A_T$ & $E_T$ \\
\hline
 \color{red}$t_{12}$ & \color{red}1 & \color{red}2 \\
 $t_{21}$ & 2 & 1 \\
 $t_{22}$ & 2 & 2
\end{tabular}
&
\begin{tabular}{@{}c@{~}||@{~}c@{~}@{~}c@{~}@{~}}
 $U$ & $E_U$ & $F$ \\
\hline
 \color{red}$u_{11}$ & \color{red}1 & \color{red}1 \\
 \color{red}$u_{21}$ & \color{red}2 & \color{red}1 \\
 \color{red}$u_{22}$ & \color{red}2 & \color{red}2
\end{tabular}
\end{tabular}
\end{small}
\end{center}
\vspace*{-1em}
\caption{An f-tree $\mathcal{T}$ and a database $\mathbf{D}$ during the execution of $\textbf{gen2}$.}
\label{fig:exec}
\vspace*{-1em}
\end{figure}

After these ranges are found in line (1), they are passed to a next
call of $\textbf{gen2}$ on the subtree formed by the children of
$\{A_R,A_S,A_T\}$. When this call returns, the next execution of line
(1) finds the ranges in $R$, $S$ and $T$ with $A_R = A_S =
A_T = 2$, the range in $U$ being again unchanged.

For further illustration, we list all recursively invoked calls
$\textbf{gen2}(\mathcal{S},\textrm{ranges})$, for which $\mathcal{S}$
is a subtree of $\mathcal{T}$ rooted at an internal node. We use
indentation to express the recursion of the calls. For brevity, we
only give the root of $\mathcal{S}$ instead of $\mathcal{S}$, and we
specify the ranges by giving the characterising condition $\gamma$.

\begin{small}
\begin{align*}
&\textbf{gen2}(A^*,\top) = (r_{111}(s_{111}+s_{112}) + r_{122}s_{121})t_{12}(u_{21}+u_{22}) + \\
& \qquad\qquad\qquad\quad + r_{212}s_{211}(t_{21}u_{11} + t_{22}(u_{21} + u_{22})) \\
&\qquad \textbf{gen2}(B^*,A^*=1) = (r_{111}(s_{111}+s_{112}) + r_{122}s_{121} \\ 
&\qquad\qquad \textbf{gen2}(C,A^*=1\wedge B^*=1) = r_{111} \\
&\qquad\qquad \textbf{gen2}(D,A^*=1\wedge B^*=1) = (s_{111}+s_{112}) \\
&\qquad\qquad \textbf{gen2}(C,A^*=1\wedge B^*=2) = r_{122} \\
&\qquad\qquad \textbf{gen2}(D,A^*=1\wedge B^*=2) = s_{121} \\
&\qquad \textbf{gen2}(E^*,A^*=1) = t_{12}(u_{21}+u_{22}) \\ 
&\qquad\qquad \textbf{gen2}(F,A^*=1\wedge E^*=2) = (u_{21}+u_{22}) \\
&\qquad \textbf{gen2}(B^*,A^*=2) = r_{212}s_{211} \\
&\qquad\qquad \textbf{gen2}(C,A^*=2,B^*=1) = r_{212} \\
&\qquad\qquad \textbf{gen2}(D,A^*=2,B^*=1) = s_{211} \\
&\qquad \textbf{gen2}(E^*,A^*=2) = t_{21}u_{11} + t_{22}(u_{21} + u_{22}) \\ 
&\qquad\qquad \textbf{gen2}(F,A^*=2\wedge E^*=1) = t_{21}u_{11} \\
&\qquad\qquad \textbf{gen2}(F,A^*=2\wedge E^*=2) = t_{22}(u_{21} + u_{22}).
\end{align*}
\end{small}\punto
\end{example} 

%%%%%%%%%%%%%%%%%%%%%%%%%%%%%%%%%%%%%%%%
\medskip%\subsection{Complexity Analysis}

We next investigate the time complexity of $\mathbf{gen2}$ as well as
the size of the produced f-representation.

The first observation is that all lines apart from line (1)
take time linear in the output size. Consider now any particular call
of $\mathbf{gen2}$ on a subtree rooted at node $A^*$ and denote by $P$
the path from the root of the f-tree to the node $A^*$. During the execution of
the loop containing line (1), for each relation $R_i \in r(A)$,
the tuples in the range $(\mathrm{start}_i,\mathrm{end}_i)$ are sorted
by their attributes in $P$ in the order they occur in $P$. Therefore
the iteration over all the maximal subranges
$(\mathrm{start}'_i,\mathrm{end}'_i)$ sharing the same value of $A^*$
can be done in a mergesort-like manner with a single simultaneous pass
of the pointers $\mathrm{start}'_i$ and $\mathrm{end}'_i$ through the
corresponding ranges $(\mathrm{start}_i,\mathrm{end}_i)$. Since we assume
that the tuples are of constant size, the time taken by line (1) is linear
in the number of tuples in these ranges $(\mathrm{start}_i,\mathrm{end}_i)$,
for those $i$ such that $R_i\in r(A)$. (For other $R_i$ we keep
$(\mathrm{start}'_i,\mathrm{end}'_i) = (\mathrm{start}_i,\mathrm{end}_i)$.)

\begin{lemma}
\label{lem:line1}
The time taken by line \emph{(1)} of ${\rm \bf gen2}$ when
computing the $\T$-factorisation of $Q(\D)$ is
$O(|Q|\cdot|\D|^{f(T)+1})$.
\end{lemma}

The time taken by the remaining lines is linear in the output
size.

\begin{lemma}
\label{lem:line2}
For any f-tree $\mathcal{T}$ of a query $Q$ and any database $\mathbf{D}$, the $\T$-factorisation of $Q(\D)$ has size at most $|\D|^{f(\T)+1}$.
\end{lemma}
\begin{proof}
By Corollary~\ref{cor:upper}, for any relation $R$, each identifier $r$ of a tuple from $R$ occurs at most $|\mathbf{D}|^{\rho^*(Q_R)} \leq |\mathbf{D}|^{f(\T)}$ times in the $\T$-factorisation of $Q(\D)$. There are at most $|\D|$ different identifiers in the $\T$-factorisation, so the total number of (occurrences of) identifiers is at most $|\D|^{f(Q)+1}$.
\end{proof}

Additionally, we need to sort the relations of $\D$ in the correct order before executing $\textbf{gen2}$, which takes time $O(|\D|\log|\D|)$. Putting this together with Lemma~\ref{lem:line1} and Lemma~\ref{lem:line2}, we obtain a bound on the total running time of our factorisation algorithm $\mathbf{gen2}$.

\begin{theorem}
\label{thm:falgo-complexity}
For any f-tree $\T$ of a query $Q$ and any database $\D$, the algorithm ${\rm
\bf gen2}$ computes the $\T$-factorisation of $Q(\D)$ in
time $O(|Q|\cdot|\D|\log|\D|)$ for hierarchical queries and $f(\T) = 0$, and
 $O(|Q|\cdot|\D|^{f(\T)+1})$ otherwise.
\end{theorem}

There is a close parallel between our results and the results of~\cite{AGM08,GM06}. They show that for a fixed query $Q$, the flat representation of $Q(\D)$ has size $O(|\D|^{\rho^*(Q)})$ and can be computed in time $O(|\D|^{\rho^*(Q)+1})$ for any database $\D$, while in Lemma~\ref{lem:line2} and Theorem~\ref{thm:falgo-complexity} we show that by allowing factorised representations, we can find one of size $O(|\D|^{f(Q)})$ in time $O(|\D|^{f(Q)+1})$.

The improvement in the exponent from $\rho^*(Q)$ to $f(Q)$ is quantified by passing from a fractional edge cover of the whole query to the fractional edge covers of the individual non-relevant parts for each relation in an optimal f-tree of $Q$. If $Q$ admits an f-tree with high degree of branching, this improvement can be substantial. There are queries for which $\rho^*(Q) = \Theta(|Q|)$ while $f(Q) = O(1)$, the simplest example being a product query of $n$ relations. For such cases, the savings in both size of the representation and the time needed to compute it are exponential in $|Q|$.

\section{Equalities with Constants}
\label{sec:constants}

We can extend our results to select-project-join queries whose selections contain equalities with constants. In the following, we call such queries simply \emph{queries with constants}.

Consider any query $Q = \pi_{\bar{A}} ( \sigma_\phi (R_1 \times \dots \times R_n))$ where $\phi$ contains equalities with constants, and without loss of generality assume that $\phi$ is satisfiable. Denote by $\C$ the set of all attributes of $Q$ which are equated in $\phi$ to constants, either directly or transitively. Let $\phi_\C$ be the conjunction of equalities from $\phi$ which involve attributes from $\C$ and let $\phi'$ be the conjunction of equalities from $\phi$ which do not involve attributes from $\C$. Then $\phi = \phi' \wedge \phi_\C$, and hence $Q = \pi_{\bar{A}} ( \sigma_{\phi'} \sigma_{\phi_\C}(R_1 \times \dots \times R_n))$.

Define the query $Q' = \pi_{\bar{A}} ( \sigma_{\phi'} (R_1 \times \dots \times R_n))$. Then for any database $\D$, we have $Q(\D) = Q'(\sigma_{\phi_{\C}}(\D))$. Since $Q'$ is now a select-project-join query without constants, this enables us to describe the factorisation properties of $Q(\D)$ using our existing results.

Let us first extend our main definitions to queries with constants.

\begin{definition}
For any query $Q$ with constants, $\T$ is called an f-tree for $Q$ if it is an f-tree for $Q'$.\punto
\end{definition}

\begin{definition}
The $\T$-factorisation of $Q(\D)$ is defined to be the $\T$-factorisa\-tion of $Q'(\sigma_{\phi_{\C}}(\D))$.\punto
\end{definition}

It follows immediately from $Q'(\sigma_{\phi_{\C}}(\D)) = Q(\D)$ that this definition is sound, i.e.\ that the $\T$-factorisation of $Q(\D)$ is indeed an f-representation of $Q(\D)$. Just as for queries without constants, we can now define $f(Q)$ to be the minimum $f(\T)$ over all f-trees $\T$ for $Q$. Equivalently, $f(Q) = f(Q')$.

\begin{corollary}[Extends Corollary~\ref{cor:upperbound}]
\label{cor:const-upper}
For any query $Q$ with constants and any database $\D$, the readability of $Q(\D)$ is at most $M\cdot |\D|^{f(Q)}$.
\end{corollary}

\begin{corollary}[Extends Corollary~\ref{cor:lowerbound}]
\label{cor:const-lower}
Let $Q$ be a query with constants. For any f-tree $\T$ of $Q$ there exist arbitrarily large databases $\D$ for which the $\T$-factorisation $Q(\D)$ is at least read-$(|\D|/|Q|)^{f(Q)}$.
\end{corollary}

The dichotomy between non-repeating queries of bounded and unbounded readability extends to queries with constants with only a slight change.

\begin{corollary}[Extends Theorem~\ref{th:unbounded}]
\label{cor:const-hierarchical}
Let $Q$ be a non-repeating query with constants. If $Q'$ is
hierarchical, then the readability of $Q(\D)$ is 1 for any database
$\D$. If $Q'$ is non-hierarchical, then there exist arbitrarily large
databases $\D$ such that the readability of $Q(\D)$ is
$\Omega(\sqrt{|\D|})$.
\end{corollary}

Since the f-trees for $Q$ are the same as for $Q'$, to enumerate the f-trees for $Q$ and to find an optimal one and hence $f(Q)$, it suffices to compute $Q'$ from $Q$ and use the existing algorithms from Section~\ref{sec:algorithm}.

Finally, to compute the $\T$-factorisation of $Q(\D)$, it is sufficient to compute $Q'$ from $Q$ and $\sigma_{\phi_\C}(\D)$ from $\D$, and then to use existing algorithms from Section~\ref{sec:falgorithm} to compute the $\T$-factorisation of $Q'(\sigma_{\phi_\C}(\D))$. Computing $Q'$ takes time $O(|Q|^2)$ and computing $\sigma_{\phi_\C}(\D)$ takes time $O(|Q|\cdot |D|)$.

\begin{corollary}[Extends Theorem~\ref{thm:falgo-complexity}]
\label{cor:const-falgo}
For any f-tree $\T$ of a query $Q$ with constants and any database $\D$, we can compute the $\T$-factorisation of $Q(\D)$ in time $O(|Q|\cdot|\D|\log|\D|+|Q|^2)$ for hierarchical queries and $f(\T) = 0$, and \mbox{$O(|Q|\cdot|\D|^{f(\T)+1}+|Q|^2)$} otherwise.
\end{corollary}

\section{Conclusion}
\label{sec:conclusion}

This work is the start of a research agenda on a new kind of
representation systems and query evaluation techniques, where the
logical model is that of relational databases yet the actual physical
model is that of factorised representations. As a necessary first
step, this paper classifies select-project-join queries based on their
worst-case result size as factorised representations. We consider bag
semantics for query evaluation here.  We plan to further study the
problems of query evaluation on factorised representations, designing
a factorisation-aware storage manager, as well as approximations of
queries with non-polynomial readability by lower and upper bound
queries with polynomial readability. We also plan to develop a
visualisation approach of query results based on factorised
representations.

%%%%%%%%%% BIBLIOGRAPHY

\bibliographystyle{alpha}
\bibliography{report}

\newpage
\appendix
\setcounter{section}{0}
\setcounter{subsection}{0}
\def\thesection{\Alph{section}}
\section{Deferred Proofs}

%%%%%%%%%%%%%%%%%%%%%%%%%%%%%%%%%%%%%%%%%%%%%%%%%%%%%%
\section*{Proofs from Section~\ref{sec:factorised-representations}}

%%%%%%%%%%%%%%%%%%%%%%%%%%%%%%%%%%%%%%%%%%%%%%%%%%%%%%
\subsection*{Proof of Proposition~\ref{prop:tuple-enumeration}}

We show that the tuples of an f-representation $\Phi$ can be
enumerated with $O(|\Phi|\log |\Phi|)$ delay and space.

Each tuple represented by the f-representation $\Phi$ corresponds to a monomial of the polynomial of $\Phi$, and each such monomial consists of the identifiers reached by recursively choosing one summand at each sum and all factors at each product.

We can use pointers to keep track of the choice of summand at each sum. In general, we may have $O(|\Phi|)$ sums, and need $O(\log|\Phi|)$ space per pointer. Any choice of pointers corresponds to a monomial of $\Phi$ obtained by recursively exploring $\Phi$, following the chosen summands and multiplying together all the reached identifiers. This can be done in time $O(|\Phi|\log|\Phi|)$ by a simple depth-first search. Not all sums are reached by this process, since some of them lie inside other summands which were not chosen. Call such sums \emph{disabled}, and call the reachable sums \emph{enabled}.

Initially, the pointer at each sum is set to the first summand of the sum. This choice of pointers defines the first monomial. Consider any order $\pi$ of the sums that is consistent with their nesting in $\Phi$, i.e.\ such that outer sums appear earlier in $\pi$ than inner sums. To advance to the next monomial, we \emph{advance} the pointer of the last enabled sum in $\pi$. Advancing the pointer of a sum $S$ consists of updating it to point to the next summand of $S$. In case it already points to the last summand, we update it back to the first summand, and recursively advance the last enabled sum preceding $S$ in $\pi$. If $S$ is already the first enabled sum, we terminate.

Updating a pointer of a sum potentially disables and enables other sums, but only sums appearing later in $\pi$. The above process proceeds backwards in $\pi$, traversing the enabled sums until finding one which is not pointing at its last summand. Therefore, when advancing to the next monomial, we can first find the enabled sums using a depth-first search in time $O(|\Phi|\log|\Phi|)$, and then use this information when advancing the pointers in time $O(|\Phi|)$. Finally, finding the next monomial using the updated pointers takes time $O(|\Phi|\log|\Phi|)$. The total delay between outputting two monomials is thus $O(|\Phi|\log|\Phi|)$.

%%%%%%%%%%%%%%%%%%%%%%%%%%%%%%%%%%%%%%%%%%%%%%%%%%%%%%
\subsection*{Proof of Lemma~\ref{lemma:readability-rst}}

 We show that the polynomial $p_N = \sum_{i,j = 1}^N r_i s_{ij} t_j$
 has readability $\frac{N}{2} + O(1)$.

  We first show that $p_N$ has readability at least $\frac{N}{2}$. Let
  $\psi$ be any polynomial equivalent to $p_N$ and consider its parse
  tree, where adjacent sum nodes are aggregated into a single node. If
  we expand $\psi$ by distributivity of product over sum, we must
  obtain the expression $p_N = \sum_{i,j = 1}^N r_i s_{ij} t_j$.
  Therefore, there must be exactly one occurrence of $s_{ij}$ in the
  parse tree, and it can have at most two multiplications on its path
  to the root. If there are two multiplications, $\psi$ is of the form
  \begin{align*}
    ((s_{ij} + \psi_1)(r_i + \psi_2) + \psi_3)(t_j + \psi_4) + \psi_5& \qquad\textrm{or} \\
    ((s_{ij} + \psi_1)(t_j + \psi_2) + \psi_3)(r_i + \psi_4) + \psi_5&,
  \end{align*}
  but then necessarily, $\psi_1$, $\psi_2$ and $\psi_4$ are empty
  (because if any two of $r_i$, $s_{ij}$, or $t_j$ appear in a monomial
  in the result, the monomial must be $r_is_{ij}t_j$). Similarly, if there
  is one multiplication, $\psi$ is of the form
  \[(s_{ij} + \psi_1)((r_i + \psi_2)(t_j + \psi_3) + \psi_4) +
  \psi_5,\] but then necessarily all of $\psi_1$, $\psi_2$, $\psi_3$
  and $\psi_4$ are empty. In any case, $s_{ij}$ appears in one of the
  forms
  \[(s_{ij}r_i + \dots)t_j + \dots \qquad\textrm{or}\qquad (s_{ij}t_j
  + \dots)r_i + \dots.\] Therefore, each $s_{ij}$ appears directly in
  a binary product with a $r$- or $t$-identifier. Since there are
  $N^2$ of the $s$-identifiers, and $2N$ different $r$- and
  $t$-identifiers, at least one of the latter occurs at least
  $\frac{N}{2}$ times in the expression~$\psi$.

  To complete the proof, it is enough to exhibit a read-$(\frac{N}{2}
  + O(1))$ factorisation of $p_N$. Defining $a_N$ and $b_N$ as
  \begin{align*}
    a_N &= \textstyle\sum_{i=1}^N \sum_{j=0}^{\lfloor N/2 \rfloor - 1} r_i s_{i(i+j)} t_{i+j} \\
    &= \textstyle\sum_{i=1}^N r_i \left( \sum_{j=0}^{\lfloor N/2 \rfloor - 1} s_{i(i+j)} t_{i+j} \right), \tag{A} \\
    b_N &= \textstyle\sum_{i=1}^N \sum_{j=\lfloor N/2 \rfloor}^{N-1} r_i s_{i(i+j)} t_{i+j} \\
    &= \textstyle\sum_{i=1}^N \sum_{j=\lfloor N/2 \rfloor}^{N-1} r_{i-j} s_{(i-j)i} t_{i} \\
    &= \textstyle\sum_{i=1}^N \left( \sum_{j=\lfloor N/2 \rfloor}^{N-1} r_{i-j} s_{(i-j)i} \right) t_{i}, \tag{B}
  \end{align*}
  where all indices are considered modulo $N$, we get $p_N = a_N + b_N$. Each $s_{ij}$ occurs once either in expression A or expression B. Each $r_i$ occurs once in A and $\lceil \frac{N}{2} \rceil$ times in B, and each $t_j$ occurs $\lfloor \frac{N}{2} \rfloor$ times in A and once in B. Thus, writing $p_N$ as the sum of expressions A and B, we get a read-$(\lceil \frac{N}{2}\rceil + 1)$ factorisation.

%%%%%%%%%%%%%%%%%%%%%%%%%%%%%%%%%%%%%%%%%%%%%
\subsection*{Proof of Theorem~\ref{th:crown}}

We show that the readability of the polynomial $q_N = \sum_{i,j=1; i\neq j}^ N a_ib_j$ is\\ $\Omega(\frac{\log{N}}{\log\log{N}})$ but $O(\log{N})$.

We first prove the lower bound. Any factorisation of the polynomial $q_N = \sum_{i,j=1; i\neq j}^ N a_ib_j$ is of the form $\sum_i A_i B_i$, where each $A_i$ is a sum of $a$-variables and each $B_i$ a sum of $b$-variables. Represent each monomial $a_ib_j$ as an edge in the bipartite graph $K_{N,N}$ and each sum of such monomials as a union of the corresponding edges. Each product $A_i B_i$ then represents a complete bipartite subgraph (called a \emph{biclique}) of $K_{N,N}$, and the factorisation $\sum_i A_i B_i$ represents an edge-disjoint union of such bicliques. 

If $\sum_i A_i B_i = q_N$, this union must be equal to the graph represented by $q_N$, that is, the crown graph $G_N = \{(a_i,b_j) \mid i\neq j\} \subset K_{N,N}$. The number of occurrences of a variable in the factorisation is the number of its bicliques containing the corresponding vertex of $K_{N,N}$.

The readability of $q_N$, denote it $\rho_N$, is then the smallest $k$ for which $G_N$ can be written as a union of bicliques in such a manner that each its vertex is included in at most $k$ of these bicliques.

Let $M_k$ be the largest $N$ for which $G_N$ can be written as a union of bicliques in such a manner that each its vertex is included in at most $k$ of these bicliques. Then $\rho_N$ is the smallest $k$ for which $M_k \geq N$.

\begin{lemma}
 $M_1 = 2$.
\end{lemma}
\begin{proof}
  On one hand, $G_2 = K_{1,1} + K_{1,1}$ can be written as a vertex-disjoint
  union of bicliques. On the other hand, $G_3$ clearly cannot be
  written as a vertex-disjoint union of bicliques.
\end{proof}

\begin{lemma}
  For $k>1$, $M_k < k^2 M_{k-1}$.
\end{lemma}
\begin{proof}
  First introduce some notation: if $A = \{a_{i_1},a_{i_2},$ $\dots,a_{i_k}\}$ is a set of vertices of $G_N$ all coming from one partition, by $\overline{A}$ we will denote the \emph{opposite set} $\{b_{i_1},b_{i_2},\dots,$ $b_{i_k}\}$. Similarly we define $\overline{B}$ for a set $B$ of vertices from the other partition.

  Let $k>1$, $N = k^2 M_{k-1}$, and suppose that $G_N$ can be written
  as a union of bicliques such that each vertex is contained in at
  most $k$ of them.
 
  Consider one such collection $\mathcal{C}$ of bicliques. The vertex
  $a_1$ is contained in at most $k$ bicliques $\{A_i \times B_i\}_i
  \subseteq \mathcal{C}$, and the vertex $b_1$ is contained in at most
  $k$ bicliques $\{A_i' \times B_i'\}_i \subseteq \mathcal{C}$. Since
  $\bigcup \mathcal{C} = G_n$, we must have
 \[\bigcup_i B_i = \{b_2,\dots,b_N\} \qquad\textrm{and}\qquad \bigcup_i A_i' = \{a_2,\dots,a_N\}.\]
 Since $|\bigcup_i B_i| = N-1 = k^2 M_{k-1} - 1$, and since $k>1$, by
 the pigeonhole principle there exists some $B_j$ such that $|B_j|
 \geq k M_{k-1}$. But $\bigcup_i A_i' = \{a_2,\dots,a_N\}$ implies
 $\bigcup_i \overline{A_i'} = \{b_2,\dots,b_N\} \supseteq B_j$, so
 there exists some $A_i'$ such that $|\overline{A_i'} \cap B_j| \geq
 M_{k-1}$. Denote $A = (A_i' \cap \overline{B_j})$. This means that $A
 \subseteq A_i'$ and $\overline{A} \subseteq B_j$. And $|A| \geq
 M_{k-1}$.
 
 Now consider the collection $\mathcal{C}$ restricted to $A \times
 \overline{A}$, i.e.
 \[\mathcal{D} = \mathcal{C} \upharpoonright_{A\times \overline{A}} = \{(X\times Y) \cap (A\times \overline{A}) \mid X\times Y \in \mathcal{C}\}. \]
 This is still a collection of bicliques, and it covers the graph
 induced by $A \times \overline{A}$, which is in fact isomorphic to
 $G_{|A|}$. Since $|A| \geq M_{k-1}$, at least one vertex $v$ of this
 subgraph is contained in at least $k$ bicliques in $\mathcal{D}$.
 Since all bicliques in $\mathcal{D}$ are restrictions of bicliques in
 $\mathcal{C}$, $v$ is also included in the corresponding bicliques
 from $\mathcal{C}$. However, $v$ is also included in the biclique
 $A_j \times B_j$ or $A_i' \times B_i'$ (depending on the partition it
 is in). But since $A_j \cap A = \emptyset$ and $\overline{A} \cap
 B_i' = \emptyset$, the restrictions of these two bicliques to
 $A\times\overline{A}$ are empty, and thus neither of them is one of
 our original $k$ bicliques containing $v$. Therefore, $v$ is in fact
 included in at least $k+1$ bicliques from $\mathcal{C}$, which is a
 contradiction to our assumption.
\end{proof}

\begin{corollary}
 For $k \geq 1$, $M_k \leq 2(k!)^2$.
\end{corollary}

\begin{lemma}
 With $k= \frac{\log{N}}{\log\log{N}}$, we have $2(k!)^2 < N$ for large enough $N$. 
\end{lemma}

\begin{corollary}
 $\rho_N = \Omega(\frac{\log{N}}{\log\log{N}})$.
\end{corollary}

For the upper bound on readability, we prove the following lemma.

\begin{lemma}
For any $N>1$, $\rho_N \leq \rho_{\lceil N/2 \rceil + 1}$.
\end{lemma}
\begin{proof}
Write $q_N$ as
\begin{align*}
 q_N &= \textstyle\sum_{i,j=1,i\neq j}^{\lceil N/2 \rceil} r_it_j +  \textstyle\sum_{i,j=\lceil N/2 \rceil + 1,i\neq j}^N r_it_j + \phantom{x}\\
&+ (\textstyle\sum_{i=1}^{\lceil N/2 \rceil} t_i)(\sum_{j=\lceil N/2 \rceil + 1}^N r_j) + \phantom{x}\\
&+ (\textstyle\sum_{i=\lceil N/2 \rceil + 1}^N t_i)(\sum_{j=1}^{\lceil N/2 \rceil} r_j).
\end{align*}
The first two sums are equivalent to $q_{\lceil N/2\rceil}$ and $q_{\lfloor N/2\rfloor}$ respectively, so they are both equivalent to at most read-$\rho_{\lceil N/2\rceil}$ expressions, but they contain different variables. In the rest of the expression, each variable appears at most once. Therefore, the whole expression is equivalent to a read-\mbox{$(\rho_{\lceil N/2\rceil}+1)$} expression. This completes the proof.
\end{proof}

\begin{corollary}
 By induction, $\rho_N = O(\log{N})$.
\end{corollary}

%%%%%%%%%%%%%%%%%%%%%%%%%%%%%%%%%%%%%%%%%%%%%%%%%%%%%%
\section*{Proofs from Section~\ref{sec:ftrees}}

%%%%%%%%%%%%%%%%%%%%%%%%%%%%%%%%%%%%%%%%%%%%%%%%%%%%%%%%%%%%%%%%%
\subsection*{Proof of Proposition~\ref{prop:correctness-ftree}}

 We show that for any f-tree $\T$ of a query $Q$, $\Phi(\T)$ is an
 f-representation of $Q(\D)$ for any database $\D$.

  To convert $\Phi(\T)$ into a sum-of-products form, we repeatedly
  choose any sum $\sum_{A^*}$ appearing inside a product and
  distribute all the other factors to each of the summands. However,
  for each attribute class $A^*$, all relations with any attribute
  from $A^*$ must appear as leaves of the subtree rooted at $A^*$, and
  hence all tuples from these relations must already appear inside the
  sum $\sum_{A^*}$.  Therefore, when moving factors into a sum
  $\sum_{A^*}$, we can also extend the conditions $A^*=a$ to these
  factors, as it will not affect the selections on relations contained
  in them.  We can then move all the products downwards, obtaining an
  expression
  \[\sum_{A_1^*}\cdots \sum_{A_n^*} \prod_R \textstyle \sum_{t_j \in
    \sigma_\gamma(\mathbf{R})} id_j\tuple{\pi_{\mbox{head}(Q)}(t_j)},\tag{1}\] where
  the sums are over all equivalence classes of attributes and the
  product over all relations of $Q$. This is equivalent to the
  sum-of-products representation of $Q(\mathbf{D})$.

%%%%%%%%%%%%%%%%%%%%%%%%%%%%%%%%%%%%%%%%%%%%%%%%%%%%%%%%%%%%%%%%%%%%

\subsection*{Proof of Proposition~\ref{prop:ftree-hierarchical}}

We show that a query is hierarchical if and only if it has an f-tree $\T$ such that $\mbox{\textrm{Non-relevant}}(R)=\emptyset$ for each relation $R$.

Let $Q$ be a hierarchical query. By Proposition~\ref{prop:algo-hierarchical}, when computing the f-tree $\T$, the algorithm \textbf{iter-pruned} only has a single choice for the root of each subtree. This means that for each node $A^*$ in the tree, all its children $B^*$ satisfy $r(A) \supseteq r(B)$. Therefore, the nodes relevant to each relation $R$ not only lie on a path from the root of $\T$, but form a contiguous path from the root of $\T$. The leaf labelled by $R$ is put directly under the lowest node of this path, and we get $\mbox{\textrm{Non-relevant}}(R)=\emptyset$.

Conversely, suppose that $\T$ is an f-tree for $Q$ such that $\mbox{\textrm{Non-relevant}}(R)=\emptyset$ for each relation $R$. For any two attribute classes $A^*$ and $B^*$ of $Q$, either one is an ancestor of the other, or they appear in sibling subtrees. In the latter case, $r(A)$ and $r(B)$ are disjoint. In the former case, suppose wlog that $A^*$ is an ancestor of $B^*$. Any relation $R \in r(B)$ must appear in a leaf under the node $B^*$. However, since $\mbox{\textrm{Non-relevant}}(R)=\emptyset$, all nodes on the path from $R$ to its root are relevant to $R$, and we must also have $R \in r(A)$. This shows that $r(B)\subseteq r(A)$ and completes the proof.

%%%%%%%%%%%%%%%%%%%%%%%%%%%%%%%%%%%%%%%%%%%%%%%%%%%%%%
\section*{Proofs from Section~\ref{sec:query-readability}}

%%%%%%%%%%%%%%%%%%%%%%%%%%%%%%%%%%%%%%%%%%%%%%%%%%%%%%%%%%%%%%%%%%%%
\subsection*{Proof of Lemma~\ref{lemma:number-occurrences}}

Let $Q = \pi_{\bar{A}}\sigma_\phi (R_1\times \dots \times R_n)$ be a
query, $\mathcal{T}$ be an f-tree of $Q$, and $\Phi({\cal T})$ be the
$\T$-factorisation of $Q(\D)$. Also let $R=R_i$ be a relation of
$Q$. By $NR$ we denote the set $\mbox{Non-relevant}(R)$ and by ${\cal
S}(R)=\tuple{t}$ we denote the conjunction of equalities of all
attributes of $R$ to corresponding values in
$\tuple{t}$. Lemma~\ref{lemma:number-occurrences} claims that for any
database $\mathbf{D}$, the number of occurrences of the identifier $r$
of a tuple $\tuple{t}$ from $R$ in $\Phi({\cal T})$ is equal to the
number of distinct tuples in
\[\big(\pi_{NR} ( \sigma_{{\cal S}(R)=\tuple{t}} \sigma_\phi (
R_1\times \dots \times R_n))\big)(\D).\]

In $\Phi(\mathcal{T})$, each time an expression ${\llbracket
\mbox{leaf } R \rrbracket(\gamma)}$ is generated from the leaf $R$, it
appears inside the summations $\sum_{A^*}$ over all the values of
attribute classes $A^*$ from $\textrm{Path}(R)$. Thus, each time an
identifier $r$ of a tuple $\langle t \rangle$ from $R$ appears in
$\Phi(\mathcal{T})$, it appears inside a ${\llbracket \mbox{leaf } R
\rrbracket(\gamma)}$ with a different condition $\gamma$ on the
attributes from $\textrm{Path}(R)$.

However, not all $\gamma$ yield the identifier $r$ in the expression ${\llbracket \mbox{leaf } R \rrbracket(\gamma)}$. Firstly, all the attributes in the nodes relevant to $R$ must be assigned the corresponding value from $\langle t \rangle$ in the condition $\gamma$, otherwise the expression will not contain the identifier $r$.

Secondly, even if the expression ${\llbracket \mbox{leaf } R \rrbracket(\gamma)}$ contains $r\tuple{t}$, it may happen that this expression is inside a product with an empty sum, and hence does not appear in the output $\Phi(\mathcal{T})$. In particular, $r\tuple{t}$ from ${\llbracket \mbox{leaf } R \rrbracket(\gamma)}$ appears in $\Phi(\mathcal{T})$ if and only if it appears in at least one monomial in the sum-of-products form of $\Phi(\mathcal{T})$. From the expanded form of $\Phi(\mathcal{T})$ given in the expression (1) in the proof of Proposition~\ref{prop:correctness-ftree}, we see that each such monomial corresponds to an extension $\gamma'$ of the condition $\gamma$ to all attribute classes, for which all other relations also give a nonempty selection.

Thus, each occurrence of $r\tuple{t}$ in $\Phi(\mathcal{T})$
corresponds to a condition $\gamma$ on the attributes from $NR$, such
that $\tuple{t}$ satisfies $\gamma$ and there exists an output tuple
of $Q(\mathbf{D})$ satisfying the condition $\gamma$. Each such
condition $\gamma$ is determined by the choice of values of the
attributes from $NR$, and each such choice of values corresponds to a
tuple in
\[\big(\pi_{NR} ( \sigma_{{\cal S}(R_i)=\tuple{t}} \sigma_\phi ( R_1\times
\dots \times R_n))\big)(\D).\]

%%%%%%%%%%%%%%%%%%%%%%%%%%%%%%%%%%%%%%%%%%%%%%%%%%%%%%%%%%%%%%%%%%%%
\subsection*{Proof of Lemma~\ref{lem:fcover-lower}}

We show that for any equi-join query $Q$, there exist arbitrarily
large databases $\D$ such that $||Q(\D)|| \geq
(|\D|/|Q|)^{\rho^*(Q)}$. We essentially repeat the proof given
in~\cite{AGM08}, fixing a minor omission and extending it to repeating
queries.

Suppose first that $Q$ is non-repeating. Denote by $a(R)$ the set of
attribute classes containing attributes of a relation $R$. The linear
program with variables $y_{A^*}$ labelled by the attribute classes of
$Q$,
\begin{align*}
\textrm{maximising}\qquad&\textstyle\sum_{A^*} y_{A^*} \\
\textrm{subject to}\qquad&\textstyle\sum_{A^* \in a(R)} y_{A^*} \leq 1 \qquad\textrm{for all relations $R$, and}\\
& y_{A^*} \geq 0 \qquad\qquad\qquad\:\:\:\,\textrm{for all $A^*$,}
\end{align*}
is dual to the program given in Definition~\ref{def:fractional-cover-number}.

By this duality, any optimal solution $\{y_{A^*}\}$ to this linear program has cost $\sum_{A^*} y_{A^*} = \rho^*(Q)$. We also know that there exists an optimal solution with rational values. Thus, there exist arbitrarily large $N$ such that $N^{y_{A^*}}$ is an integer for all $A^*$.

For any such $N$, we can construct a database $\D$ as follows. For each attribute class $A$, let $N_A = N^{y_{A^*}}$, and let $[N_A] = \{1,\dots,N_A\}$ be the domain for the attributes in $A^*$. For each relation $R$ of $Q$, let the relation instance $\mathbf{R}$ contain all tuples $t$ for which $t(A) \in [N_A]$ for all attributes $A$, but $t(A) = t(B)$ for any attributes $A$ and $B$ equated in $Q$ (i.e. such that $A^* = B^*$). For each attribute class $A^*$ in $a(R)$ there are $N_A$ possible values of the attributes in $A^*$, so the size of $\mathbf{R}$ will be
\[|\mathbf{R}| = \textstyle\prod_{A^*\in a(R)} N_A = \textstyle\prod_{A^*\in a(R)} N^{y_{A^*}} = N^{\sum_{A^*\in a(R)} y_{A^*}} \leq N.\]
This implies that $|\D| \leq |Q|\cdot N$. However, we have $\sum_{A^*\in a(R)} y_{A^*} = 1$ for at least one relation $R$ (otherwise we could increase any $y_{A^*}$ to produce a better solution to the linear program), so $|\D| \geq N$.

Any tuple $t$ in the result $Q(\D)$ is given by its values $t(A_1) =
\dots = t(A_k) \in [N_{A_1}]\}$ for each attribute class
$A_1^* = \{A_1,\dots,A_k\}$, and any such combination of values gives
a valid tuple in the output. The size of the output is thus
\[|Q(\D)| = \textstyle\prod_{A^*} N_A = N^{\sum_{A^*} y_{A^*}} = N^{\rho^*(Q)} \geq (|\D|/|Q|)^{\rho^*(Q)}.\]
Since all tuples in each relation are distinct, all tuples in the
output are distinct, and we also have $||Q(\D)|| = |Q(\D)| \geq
(|\D|/|Q|)^{\rho^*(Q)}$. The outer projection of $Q$ does not reduce
the cardinality of $Q$'s result, since we consider bag semantics.

Now suppose that $Q$ is repeating, that is, contains multiple
relations mapping to the same name. In that case, such relations
require the same relation instance as their interpretation, while the
database $\D$ constructed in the above proof may assign them different
relation instances. However, consider the database $\D'$ constructed
as follows. For any class $\{R_1,\dots,R_k\}$ of relations mapping to
the same name $R$, replace the relation instances
$\mathbf{R}_1,\dots,\mathbf{R}_k$ in $\D$ by a single relation
instance $\mathbf{R} = \bigcup_i \mathbf{R}_i$ in $\D'$.

Firstly, we have $|\D'| \leq |\D|$, since $|\bigcup_i \mathbf{R}_i| \leq \sum_i |\mathbf{R}_i|$. Secondly, we still have $|\D'| \geq N$, since the size of the largest relation in $\D'$ is at least the size of the largest relation in $\D$. Finally, we have $Q(\D') \supseteq Q(\D)$, because for any relation symbol $R_i$ of $Q$, its interpretation $\mathbf{R}$ in $\D'$ is a superset of its interpretation $\mathbf{R}_i$ in $\D$. Thus we get
\[||Q(\D')|| \geq ||Q(\D)|| \geq (|\D|/|Q|)^{\rho^*(Q)} \geq (|\D'|/|Q|)^{\rho^*(Q)},\]
which completes the proof.

%%%%%%%%%%%%%%%%%%%%%%%%%%%%%%%%%%%%%%%%%%%%%%%%%%%%%%%%%%%%%%%%%%%%
\subsection*{Proof of Lemma~\ref{lem:lower}}

Let $Q = \pi_{\bar{A}}(\sigma_\phi (R_1\times \dots \times R_n))$ be a
query, $\mathcal{T}$ be an f-tree of $Q$, and $\Phi({\cal T})$ be the
$\T$-factorisation of $Q(\D)$.  Also let $R$ be a relation in $Q$. We
show that there exist arbitrarily large databases $\mathbf{D}$ such
that each identifier $r$ from $R$ occurs in $\Phi(\T)$ at least
$(|\D|/|Q|)^{\rho^*(Q_R)}$ times.

Recall that the query $Q_R$ is obtained by restricting $Q$ to the attributes of
$NR = \textrm{Non-relevant}(R)$, and omitting the projection $\pi_{\bar{A}}$.

Applying Lemma~\ref{lem:fcover-lower} to the query $Q_R$, we obtain that there exist arbitrarily large databases $\mathbf{D}_R$ such that $||Q_R(\mathbf{D}_R)|| \geq (|\D_R|/|Q_R|)^{\rho^*(Q_R)}$. Construct the database $\mathbf{D}$ by extending $\mathbf{D}_R$: for each new attribute $A$ allowing a single value $1$, and extending each tuple in each relation by this value in the new attributes. For relations appearing in $Q$ but with no attributes in $Q_R$, the relation instance in $\D$ will consist of a single tuple with value $1$ in each attribute. Notice that $|Q_R| \leq |Q|$ and $|\D| = |\D_R|$, so that $||Q_R(\mathbf{D}_R)|| \geq (|\D|/|Q|)^{\rho^*(Q_R)}$.

Finally, a tuple from $\mathbf{D}_R$ satisfies $\phi_R$ if and only if the corresponding extended tuple satisfies $\phi$, since the values in all attributes outside $NR$ are equal. Moreover, since $R$ has no attributes in $NR$, each identifier $r$ from $R$ corresponds to the tuple $\tuple{t} = \tuple{1,\dots,1}$, and each tuple from $(R_1\times \dots \times R_n)$ satisfies $\sigma_{\mathcal{S}(R) = \tuple{t}}$. By Lemma~\ref{lemma:number-occurrences}, the number of occurrences of any $r$ from $R$ in the $\T$-factorisation of $Q(\D)$ is
\begin{align*}
&|| \pi_{NR} ( \sigma_{{\cal S}(R)=\tuple{t}} \sigma_\phi ( R_1\times
  \dots \times R_n)) || \\
  = &|| \pi_{NR} ( \sigma_\phi ( R_1\times \dots \times R_n)) || \\
  = &|| \sigma_{\phi_R} ( \pi_{NR} ( R_1\times \dots \times R_n)) || \\
  = &|| Q_R(\mathbf{D}_R)|| \\
  \geq& (|\D|/|Q|)^{\rho^*(Q_R)}. \qedhere
\end{align*}

%%%%%%%%%%%%%%%%%%%%%%%%%%%%%%%%%%%%%%%%%%%%%%%%%%%%%%%%%%%%%%%%%%%%
\subsection*{Proof of Corollary~\ref{cor:hierar}}

We show that if $Q$ is hierarchical, the readability of $Q$ is bounded by a constant, while if $Q$ is non-hierarchical, for any f-tree $\T$ of $Q$ there exist databases $\D$ such that the $\T$-factorisation of $Q(\D)$ is read-$\Theta(|\D|)$.

By Proposition~\ref{prop:ftree-hierarchical}, if $Q$ is hierarchical, there exists an f-tree $\T$ of $Q$ such that $\mbox{Non-relevant}(R) = \emptyset$ for all relations $R$. For any such tree $\T$ we have $f(\T) = 0$, hence $f(Q) = 0$, and by Theorem~\ref{th:characterisation}, the readability of $Q(\D)$ is $O(1)$.

If $Q$ is non-hierarchical, for any f-tree $\T$ there is a relation $R$ such that $\mbox{Non-relevant}(R)$ is nonempty. Then the query $Q_R$ contains at least one attribute, and hence $\rho^*(Q_R) \geq 1$. Therefore $f(\T) \geq 1$ and also $f(Q) = 1$. The result then follows from Theorem~\ref{th:characterisation}.

%%%%%%%%%%%%%%%%%%%%%%%%%%%%%%%%%%%%%%%%%%%%%%%%%%%%%%%%%%%%%%%%%%%%
\subsection*{Proof of Theorem~\ref{th:unbounded}}

We show that for a fixed non-repeating query $Q$, the following holds. If $Q$ is hierarchical, the readability of $Q(\D)$ is 1 for any database $\D$. If $Q$ is non-hierarchical, there exist arbitrarily large databases $\D$ such that the readability of $Q(\D)$ is $\Omega(\sqrt{|\D|})$.

In case $Q$ is hierarchical, then by Proposition~\ref{prop:ftree-hierarchical}, there exists an f-tree such that $\textrm{Non-relevant}(R)=\emptyset$ for any relation $R$ of $Q$. By Lemma~\ref{lemma:number-occurrences} it follows that hierarchical queries admit f-representations with readability 1.

If $Q$ is not hierarchical, there exist attribute classes $A^*$ and $B^*$ such that $r(A) \not\subseteq r(B)$, $r(B) \not\subseteq r(A)$ and $r(A) \cap r(B) \neq \emptyset$. Thus there must exist a relation $S$ with attributes from $A^*$ and $B^*$, a relation $R$ with attributes from $A^*$ but not $B^*$, and a relation $T$ with attributes from $B^*$ but not $A^*$.

Fix any positive integer $N$. Consider a database instance $\D$ in which the domains of attributes in $A^*$ and $B^*$ are $\{1,\ldots,N\}$ and the domains of all other attributes are $\{1\}$. For each relation $R$, let its interpretation $\mathbf{R}$ be the set of all possible tuples with the above domains, which respect the equivalence classes of attributes. We annotate the tuple in $R$ with $A^*$-value $i$ by $r_i$, tuple in $T$ with $B^*$-value $j$ by $t_j$, and tuple in $S$ with $A^*$-value $i$ and $B^*$-value $j$ by $s_{ij}$. All relations contain $N^2$, $N$, or only one tuple, depending on whether they contain attributes from $A^*$, $B^*$, both or none. Thus, $|\D| = \Theta(N^2)$.

The polynomial of the flat f-representation of $Q(\D)$, restricted to the identifiers from $R$, $S$ and $T$, is $\sum_{i,j=1}^N r_is_{ij}t_j$, which is exactly the polynomial $p_N$ defined in Lemma~\ref{lemma:readability-rst}. By Lemma~\ref{lemma:readability-rst}, this polynomial has readability $\Omega(N)$. Since any f-representation of $Q(\D)$ restricted to the identifiers of $R$, $S$ and $T$ is equivelant to $p_N$, $Q(\D)$ also has readability $\Omega(N) = \Omega(\sqrt{|\D|})$.

%%%%%%%%%%%%%%%%%%%%%%%%%%%%%%%%%%%%%%%%%%%%%%%%%%%%%%
\section*{Proofs from Section~\ref{sec:algorithm}}

%%%%%%%%%%%%%%%%%%%%%%%%%%%%%%%%%%%%%%%%%%%%%%%%%%%%%%%%%%%%%%%%%%%%
\subsection*{Proof of Lemma~\ref{lem:dec-path}}

For any f-tree $\T$ and relation $R$ labelling a leaf of $\T$, denote by $\textrm{Path}_{\T}(R)$ the set of ancestor nodes of $R$ in $\T$ (thus emphasising the role of the tree $\T$ in our previous notation $\textrm{Path}(R)$), and similarly for $\textrm{Non-relevant}_{\T}(R)$. We show that for any two f-trees $\T_1$ and $\T_2$ for a query $Q$, if $\textrm{Path}_{\T_1}(R) \subseteq \textrm{Path}_{\T_2}(R)$ for any relation $R$ of $Q$, then $f(\T_1) \leq f(\T_2)$.

For any relation $R$ of $Q$, if $\textrm{Path}_{\T_1}(R) \subseteq \textrm{Path}_{\T_2}(R)$, then also\\ $\textrm{Non-relevant}_{\T_1}(R) \subseteq \textrm{Non-relevant}_{\T_2}(R)$. If we let $Q_R^{\T_1}$ be the query induced by $\textrm{Non-relevant}_{\T_1}(R)$, and $Q_R^{\T_2}$ the query induced by $\textrm{Non-relevant}_{\T_2}(R)$, $Q_R^{\T_1}$ is an induced subquery of $Q_R^{\T_2}$, i.e.\ the hypergraph of $Q_R^{\T_1}$ is an induced subhypergraph of $Q_R^{\T_2}$.

If we denote by $L_1$ the fractional-cover linear program for $Q_R^{\T_1}$, as defined in Definition~\ref{def:fractional-cover-number}, and by $L_2$ the fractional-cover linear program for $Q_R^{\T_2}$, then the variables of $L_1$ are just a subset of variables of $L_2$, and the linear conditions of $L_1$ are respective restrictions of the conditions of $L_2$. Thus, any optimal solution of $L_2$ can be restricted to a feasible solution of $L_1$. The cost of such a restricted solution in $L_1$ is always at most the cost of the original solution in $L_2$, which implies that $\rho^*(Q_R^{\T_1}) \leq \rho^*(Q_R^{\T_2})$. By minimising over $R$, we obtain $f(\T_1) \leq f(\T_2)$.

%%%%%%%%%%%%%%%%%%%%%%%%%%%%%%%%%%%%%%%%%%%%%%%%%%%%%%%%%%%%%%%%%%%%
\subsection*{Proof of Lemma~\ref{lem:swap}}

Let $\T$ be an f-tree. For two nodes $A^*$ and $B^*$, we show that if $r(B) \subset r(A)$ and $B^*$ is an ancestor of $A^*$, then by swapping them we do not violate the condition ${\cal C}$ and do not increase $f(\T)$.

For any relation $R \notin r(A)$, the positions of nodes from $\mbox{Relevant}(R)$ remain unchanged. For any relation $R \in r(A)$, the leaf labelled by $R$ is under $A^*$ and hence by swapping $A^*$ and $B^*$, all nodes relevant to $R$ stay on the path from $R$ to the root. Therefore, the condition $\mathcal{C}$ remains satisfied.

It remains to prove that by this swap, the parameter $f(\T)$ does not increase. The only relations $R$ for which the set $\textrm{Path}(R)$ changes (and thus $\rho^*(Q_R)$ can change), are those lying in the subtree under $B^*$ but not in the subtree under $A^*$. For such $R$, we replace the node $B^*$ in $\textrm{Path}(R)$ by the node $A^*$. Consider the fractional-cover linear program for $Q_R$, defined as
\begin{align*}
\textrm{minimise}\qquad&\textstyle\sum_i x_i \\
\textrm{subject to}\qquad&\textstyle\sum_{i : R_i \in r(A)} x_i \geq 1 \qquad\textrm{for all attributes $A$, and}\\
& x_i \geq 0 \qquad\qquad\qquad\:\:\:\,\textrm{for all $i$.}
\end{align*}
in Definition~\ref{def:fractional-cover-number}. By replacing $B^*$ with $A^*$, the only change to this program is the strenghtening of the condition $\sum_{i : R_i \in r(B)} x_i \geq 1$ to $\sum_{i : R_i \in r(A)} x_i \geq 1$. Therefore, the cost $\rho^*(Q_R)$ of the optimal solution can only decrease. By minimising over all relations $R$ of $Q$, we conclude that 
$f(\T)$ can also only decrease.

%%%%%%%%%%%%%%%%%%%%%%%%%%%%%%%%%%%%%%%%%%%%%%%%%%%%%%%%%%%%%%%%%%%%
\subsection*{Proof of Proposition~\ref{prop:algo-hierarchical}}

We show that for a hierarchical query $Q$, the algorithm
\textbf{iter-pruned} has exactly one choice at each recursive call,
and outputs a single reduced f-tree in polynomial time.

The standard algorithm for recognising hierarchical queries (described
in~\cite{DS2007b}, though in the language of conjunctive queries) is
as follows.

\begin{compactitem}
\item Find the connected components of the query, in the sense that two relation symbols are connected if some of their attributes are equated by the query.
\item For each connected component, there must exist an attribute class with attributes in each relation in the component. If not, the query is not hierarchical. Create a node labelled by this attribute class, make it the root of an f-tree, and recurse on the rest of the component to produce its children subtrees. 
\item Output the disjoint union of the trees produced for each component.
\end{compactitem}

The connected components of the query correspond to the finest partition $P_1,\dots,P_n$ of the attribute classes such that each relation only has attributes from one $P_i$. If the considered query is hierarchical, for each such $P_i$ there exists an attribute class with attributes in each relation of $P_i$. That is, there exists at least one $A^* \in P_i$ such that for other classes $B^* \in P_i$, $r(A) \supseteq r(B)$. The lexicographically greatest such $A^*$ will be the maximum element in the $>$-order. The algorithm \textbf{iter-pruned} will therefore only consider this $A^*$ for the root of the subtree formed by $P_i$.

We have thus shown that for hierarchical queries, \textbf{iter-pruned} essentially follows the recognising algorithm given above, never branching when picking the root node, and hence outputting a single reduced f-tree. This also means that there are at most linearly many recursive calls of \textbf{iter-pruned}. Since each call takes polynomial time, the total running time is also polynomial (in the size of the query).

%%%%%%%%%%%%%%%%%%%%%%%%%%%%%%%%%%%%%%%%%%%%%%%%%%%%%%
\section*{Proofs from Section~\ref{sec:falgorithm}}

%%%%%%%%%%%%%%%%%%%%%%%%%%%%%%%%%%%%%%%%%%%%%%%%%%%%%%%%%%%%%%%%%%%%
\subsection*{Proof of Lemma~\ref{lem:line1}}

We show that the total amount of time taken by line (1) of ${\rm \bf gen2}$ when computing the $\T$-factorisation of $Q(\D)$ is $O(|Q|\cdot|\D|^{f(\T)+1})$.

Let $A^*$ be any node in $\T$, let $\mathcal{U}$ be the subtree of $\T$ rooted at $A^*$ and let $\mathrm{Path}(A)$ be the set of ancestor nodes of $A^*$. Consider any call $\textbf{gen2}(\mathcal{U},\mathcal{R})$, where $\mathcal{R}$ is a collection of ranges of tuples in $\D$. For each such call, the tuples in $\mathcal{R}$ agree on the values of attributes from $\mathrm{Path}(A)$, moreover, the ranges $\mathcal{R}$ contain all tuples of $\D$ with these values. Denote by $\gamma$ the condition on the attributes from $\mathrm{Path}(A)$ with the values given by tuples in $\mathcal{R}$. For each call $\textbf{gen2}(\mathcal{U},\mathcal{R})$, the ranges $\mathcal{R}$ are different and hence this condition $\gamma$ is different. Conversely, for any $\gamma$ such that the corresponding ranges in the relations of $\D$ are all nonempty, $\textbf{gen2}$ will be called with these ranges in the second parameter and $\mathcal{U}$ in the first parameter.

We will now calculate the total amount of time taken by line (1) in all calls of $\textbf{gen2}(\mathcal{U},\mathcal{R})$ for a fixed $\mathcal{U}$, rooted at $A^*$. We have argued before the statement of the Lemma that the amount of time taken by line (1) in any single call $\textbf{gen2}(\mathcal{U},\mathcal{R})$ is linear in the number of tuples in the ranges~$\mathcal{R}$. Instead of summing the number of tuples in $\mathcal{R}$ for each such call, we will fix a tuple $\tuple{t}$ and find the number of calls for which $\mathcal{R}$ contains this tuple. Equivalently, we will find the number of the corresponding conditions~$\gamma$ satisfied by $\tuple{t}$.

For a condition $\gamma$, the ranges $\mathcal{R}$ corresponding to $\gamma$ in $\D$ are nonempty iff $\big(\sigma_\gamma(R_1 \times \dots \times R_n)\big)(\D)$ is nonempty. Furthermore, the corresponding ranges are nonempty \emph{and} $\gamma$ is satisfied by $\tuple{t}$ iff $\big(\sigma_{\mathcal{S}(R)=\tuple{t}}(\sigma_\gamma(R_1 \times \dots \times R_n))\big)(\D)$ is nonempty. Equivalently, this is true iff $\big(\pi_{\mathrm{Path}(A)}(\sigma_{\mathcal{S}(R)=\tuple{t}}(\sigma_\gamma(R_1\times \dots \times R_n)))\big)(\D)$ is nonempty, but moreover, in such case that set contains precisely one element, which uniquely corresponds to the condition $\gamma$. Therefore, the total number of conditions $\gamma$ on the attributes of $\mathrm{Path}(A)$, for which the corresponding ranges are nonempty, and which are satisfied by $\tuple{t}$, is
\begin{align*}
&\phantom{=} \textstyle\sum_{\gamma} ||\big(\pi_{\mathrm{Path}(A)}(\sigma_{\mathcal{S}(R)=\tuple{t}}(\sigma_\gamma(R_1 \times \dots \times R_n)))
\big)(\D)|| \\
&= ||\textstyle\bigcup_{\gamma} \pi_{\mathrm{Path}(A)}(\sigma_{\mathcal{S}(R)=\tuple{t}}(\sigma_\gamma(R_1 \times \dots \times R_n)))\big)(\D)|| \\
&= ||\big(\pi_{\mathrm{Path}(A)}(\sigma_{\mathcal{S}(R)=\tuple{t}}(\sigma_\alpha(R_1 \times \dots \times R_n)))\big)(\D)||,
\end{align*}
where $\sum_\gamma$ and $\bigcup_\gamma$ range over all possible conditions $\gamma$ assigning values from $\D$ to attribute classes of $\mathrm{Path}(A)$, and $\alpha$ expresses the equality of attributes in each attribute class of $\mathrm{Path}(A)$, without assigning them particular values. However, if we let $NA = \mathrm{Path}(A) \setminus \mathrm{Relevant}(R)$, we get 
\begin{align*}
&\phantom{=} ||\big(\pi_{\mathrm{Path}(A)}(\sigma_{\mathcal{S}(R)=\tuple{t}}(\sigma_\alpha(R_1 \times \dots \times R_n)))\big)(\D)||\\
&= ||\big(\pi_{NA}(\sigma_{\mathcal{S}(R)=\tuple{t}}(\sigma_\alpha(R_1 \times \dots \times R_n)))\big)(\D)||\\
&\leq ||\big(\pi_{NA}(\sigma_\alpha(R_1 \times \dots \times R_n))\big)(\D)||\\
&\leq ||\big(\sigma_{\phi_{NA}}(\pi_{NA}(R_1 \times \dots \times R_n))
\big)(\D)||\\
&= ||Q_{NA}(\D_{NA})||,
\end{align*}
where $Q_{NA}$ and $\D_{NA}$ are defined analogously to $Q_R$ and
$\D_R$. By Lemma~\ref{lem:fcover-upper}, this number is at most
$|\D_{NA}|^{\rho^*(Q_{NA})} = |\D|^{\rho^*(Q_{NA})}$. However, since
$Q_{NA}$ is an induced subquery of $Q_R$, we have $\rho^*(Q_{NA}) \leq
\rho^*(Q_R)$, which is in turn at most $f(\T)$. We can thus conclude
that for a fixed tuple $\tuple{t}$ from a relation $R \in r(A)$, the
total number of conditions $\gamma$ on the attributes of
$\mathrm{Path}(A)$, for which the corresponding ranges are nonempty,
and which are satisfied by $\tuple{t}$, is at most $|\D|^{f(\T)}$.

There are at most $|\D|$ tuples in the relations of $r(A)$, so the total amount of time taken by line (1) in all calls of $\textbf{gen2}(\mathcal{U},\mathcal{R})$, for $\mathcal{U}$ rooted at a fixed node $A^*$, is linear in $|\D|^{f(\T)+1}$. Since there are at most $|Q|$ different nodes $A^*$, so the total time taken by line (1) is linear in $|Q|\cdot|\mathbf{D}|^{f(\mathcal{T})+1}$.

%%%%%%%%%%%%%%%%%%%%%%%%%%%%%%%%%%%%%%%%%%%%%%%%%%%%%%
\section*{Proofs from Section~\ref{sec:constants}}

%%%%%%%%%%%%%%%%%%%%%%%%%%%%%%%%%%%%%%%%%%%%%%%%%%%%%%%%%%%%%%%%%%%%
\subsection*{Proof of Corollary~\ref{cor:const-upper}}

We show that for any query $Q$ with constants and any database $\D$, the readability of $Q(\D)$ is at most $M\cdot |\D|^{f(Q)}$.

Recall that $M$ is the maximal number of relations of $Q$ mapping to the same name, and is the same for $Q$ as for $Q'$. By Corollary~\ref{cor:upperbound}, the readability of $Q(\D) = Q'(\sigma_{\phi_{\mathcal{C}}}(\D))$ is at most $M\cdot |\sigma_{\phi_{\mathcal{C}}}(\D)|^{f(Q')}$. Since $f(Q') = f(Q) \geq 0$ and $|\sigma_{\phi_{\mathcal{C}}}(\D)| \leq |\D|$, this is at most $M \cdot |\D|^{f(Q')}$.

%%%%%%%%%%%%%%%%%%%%%%%%%%%%%%%%%%%%%%%%%%%%%%%%%%%%%%%%%%%%%%%%%%%%
\subsection*{Proof of Corollary~\ref{cor:const-lower}}

We show that for any query $Q$ with constants and any f-tree $\T$ of $Q$, there exist arbitrarily large databases $\D$ for which the $\T$-factorisation $Q(\D)$ is at least read-$(|\D|/|Q|)^{f(Q)}$.

The attributes in $\C$ do not appear in any equalities in $Q'$, so each attribute is only relevant to one relation. In any f-tree of $Q'$, we can move these attributes downwards towards their respective relations, thus only decreasing the non-relevant sets for other relations, and hence not increasing $f(\T)$. It follows that there exists an f-tree $\T$ with $f(\T) = f(Q')$, such that for any relation $R$, $\mbox{Non-relevant}(R)$ does not contain any attributes from $\C$.

Now by Corollary~\ref{cor:lowerbound}, there exists arbitrarily large databases $\D$ for which the $\T$-factorisation $Q'(\D)$ is at least read-$(|\D|/|Q'|)^{f(Q')}$. Moreover, from the proof of Lemma~\ref{lem:lower} it follows that $\D$ can be constructed in such a way that for some relation $R$, all attributes not in $\mbox{Non-relevant}(R)$ have domain of size one. In particular, all attributes from $\C$ have domain of size one. By renaming the values of these attributes to the respective constants from $\phi_\C$ we can arrange that $\sigma_{\phi_{\mathcal{C}}}(\D) = \D$. Since $|Q| = |Q'|$ and $f(Q) = f(Q')$, it follows that the $\T$-factorisation $Q(\D) = Q'(\sigma_{\phi_\C}(\D))$ is at least read-$(|\D|/|Q|)^{f(Q)}$.

%%%%%%%%%%%%%%%%%%%%%%%%%%%%%%%%%%%%%%%%%%%%%%%%%%%%%%%%%%%%%%%%%%%%
\subsection*{Proof of Corollary~\ref{cor:const-hierarchical}}

Let $Q$ be a non-repeating query with constants. We show that if $Q'$ is hierarchical, the readability of $Q(\D)$ is 1 for any database $\D$, and if $Q'$ is non-hierarchical, there exist arbitrarily large databases $\D$ such that the readability of $Q(\D)$ is $\Omega(\sqrt{|\D|})$.

For hierarchical queries we have $f(Q) = 0$ and the result follows from Corollary~\ref{cor:const-upper}. For non-hierarchical queries, by Theorem~\ref{th:unbounded} there exist arbitrarily large databases $\D$ such that the readability of $Q'(\D)$ is $\Omega(\sqrt{|\D|})$. Moreover, from the proof of Theorem~\ref{th:unbounded} it follows that apart from two attribute classes $A^*$ and $B^*$ such that $r(A) \not\subseteq r(B)$, $r(B) \not\subseteq r(A)$ and $r(A) \cap r(B) \neq \emptyset$, we can arrange that all attributes have domains of size one. We cannot have $A\in\C$ or $B\in\C$, since each attribute in $\C$ is only relevant to one relation, so we can in fact arrange that all attributes from $\C$ have domains of size one. Again by simple renaming of values, we obtain $\D = \sigma_{\phi_\C}(\D)$, and hence the readability of $Q(\D) = Q'(\sigma_{\phi_\C}(\D))$ is $\Omega(\sqrt{|\D|})$.

\end{document}